\documentclass[letterpaper,twocolumn,10pt]{article}
\PassOptionsToPackage{dvipsnames}{xcolor}
\usepackage{xurl}
\usepackage{usenix-2020-09}

\usepackage{algorithm}
\usepackage{algorithmicx}
\usepackage{amsmath}
\usepackage{amsthm}
      
\usepackage{amssymb}
\usepackage{array}
\usepackage{arydshln}
\usepackage[noend]{algpseudocode}
\usepackage{boxedminipage}
\usepackage{booktabs}
\usepackage{enumitem}
\usepackage{etoolbox}
\usepackage{float}
\usepackage{graphicx}
\usepackage{hyperref}
\usepackage{listings}
\usepackage{makecell}
\usepackage{multirow}
\usepackage{pifont}
\usepackage{setspace}
\usepackage{textcomp}
\usepackage{threeparttable}
\usepackage{url}
\usepackage{subcaption}
\usepackage{tikz,balance}
\usepackage[dvipsnames]{xcolor}
\usepackage{xspace}
\usepackage{cleveref}
\Crefformat{section}{\S~#2#1#3} 
\Crefname{appendix}{Appendix}{Appendices} 
\usepackage[normalem]{ulem}

\usetikzlibrary{shapes.geometric, arrows.meta, positioning, calc, backgrounds, fit}

\hypersetup{
	colorlinks=true,
	linkcolor=ForestGreen,
	citecolor=blue,
	filecolor=magenta,
	urlcolor=.,
	pdftitle={Overleaf Example},
	pdfpagemode=FullScreen,
}

\makeatletter
\patchcmd{\@makecaption}
{\scshape}
{}
{}
{}
\makeatother

\makeatletter
\newcommand{\removelatexerror}{\let\@latex@error\@gobble}

\theoremstyle{plain}

\newtheorem{definition}{Definition}
\newtheorem{theorem}{Theorem}
\newtheorem{lemma}{Lemma}
\newtheorem{corollary}{Corollary}
\newtheorem{proposition}{Proposition}
\algnewcommand\algorithmicforeach{\textbf{for each}}
\algdef{S}[FOR]{ForEach}[1]{\algorithmicforeach\ #1\ \algorithmicdo}

\algnewcommand\algorithmicinput{\textbf{INPUT:}}
\algnewcommand\INPUT{\item[\algorithmicinput]}
\algnewcommand\algorithmicoutput{\textbf{OUTPUT:}}
\algnewcommand\OUTPUT{\item[\algorithmicoutput]}
\algnewcommand\algorithmicassume{\textbf{ASSUME:}}
\algnewcommand\ASSUME{\item[\algorithmicassume]}

\setlist[itemize]{leftmargin=*,itemsep=1pt,topsep=2pt,parsep=0pt}
\setlist[enumerate]{leftmargin=*,itemsep=1pt,topsep=2pt,parsep=0pt}

\newcommand{\ignore}[1]{}

\algrenewcommand\textproc{}
\algdef{SE}[EVENT]{Event}{EndEvent}[1]{\textbf{upon}\ #1\ \algorithmicdo}{\algorithmicend\ \textbf{event}}%
\algtext*{EndEvent}
\newcommand{\para}[1]{\vskip0.0em\noindent\textbf{#1}}

\definecolor{yescolor}{HTML}{026378}

\makeatletter
\renewcommand{\ALG@name}{}
\makeatother

\definecolor{crownGold}{HTML}{FFC107}

\newcommand{\tightdisplay}[1]{%
  \begingroup
  \setlength{\abovedisplayskip}{4pt}%
  \setlength{\belowdisplayskip}{4pt}%
  \setlength{\abovedisplayshortskip}{2pt}%
  \setlength{\belowdisplayshortskip}{2pt}%
  \[
    #1
  \]
  \endgroup
}

\newif\ifpublish
\publishtrue

\newcommand{\sysname}{AgentValve\xspace}
\newcommand{\scname}{AgentValve\xspace}
\newcommand{\network}{\mathcal{G}\xspace}

\ifpublish
    \author{
        \textnormal{Jianting Zhang} \\
        \textnormal{Purdue University} \\
        \textnormal{zhan4674@purdue.edu}
        \and
        \textnormal{Zhongtang Luo} \\
        \textnormal{Google LLC} \\
        \textnormal{zhtluo@google.com}
        \and
        \textnormal{Aniket Kate} \\
        \textnormal{Mysten Labs / Purdue University} \\
        \textnormal{aniket@purdue.edu}
    }
\fi

\date{}

\title{\Large \bf \sysname: Sybil-Resilient Payment Authorization for Permissionless AI Markets}

\begin{document}

\maketitle

\begin{abstract}
AI agents increasingly select and pay service providers without human review. In emerging permissionless markets, however, adversaries can create inexpensive identities and fabricate feedback to inflate provider reputation, making reputation-based payment authorization vulnerable to Sybil attacks and exposing clients to unbounded financial loss.

This paper presents \sysname, a protocol for Sybil-resilient payment authorization with bounded exposure in permissionless AI markets. \sysname builds on a trust network that represents reputation as finite capacity in a directed graph and authorizes payments through trust paths with sufficient available capacity. In \sysname, each admitted payment atomically reserves the capacity, with which the honest agents' payment exposure is bounded by the capacity they extend into the adversarial region. As a result, creating identities or interactions entirely within the adversarial region cannot increase this bound, making the payment authorization Sybil-resilient. To limit capacity depletion by malicious senders, \sysname combines payment-backed capacity creation with bounded route length, preventing attackers from depleting honest capacity without incurring corresponding payment costs. \sysname further decouples off-chain route discovery from on-chain verification and introduces a witness-commit mechanism that reduces route verification to linear work while requiring only constant persistent state per admitted payment.

We implement a \sysname prototype and evaluate it using six months of Ethereum data and synthetic workloads. Experimental results show that under attacks with up to 10,000 synthetic Sybil identities and with chosen attackers, the payment exposure and capacity depletion remain bounded. Regarding practicality, the witness-commit mechanism lowers admission cost by up to 40.8\% compared with baselines, and admitting common payments costs approximately \$0.04--\$0.21.
\end{abstract}

\section{Introduction} \label{sec-intro}
Large language models (LLMs) enable increasingly capable artificial intelligence (AI) agents to operate in service markets, autonomously selecting and paying service providers without human review of individual counterparties.
This emerging ecosystem is already seeing substantial adoption: A2A has gained support from more than 100 technology companies~\cite{a2a_linux}, x402 processed over 160 million agentic payments in the year preceding June 2026~\cite{x402_160m}, and more than 550,000 AI-agent identities have been registered in blockchain-based AI markets by October 2026~\cite{8004-agents}.
Emerging protocols, including A2A~\cite{a2a}, x402~\cite{x402}, and ERC-8004~\cite{erc8004}, are driving AI service markets toward a \emph{permissionless} model that promotes efficiency and autonomy, allowing agents to register and transact freely without authorization from a central operator. 

Unfortunately, permissionless AI markets introduce a significant attack surface for reputation manipulation. An adversary can create many inexpensive identities under its control, known as Sybils, and coordinate their feedback to inflate a provider's reputation~\cite{douceur2002sybil,yu2006sybilguard,yu2008sybillimit,danezis2009sybilinfer}. As a result, reputation-based service provider selection and payment authorization become vulnerable to Sybil attacks. More seriously, if payment amounts scale with a provider's reputation, an adversary can exploit reputation manipulation to induce unbounded financial losses for clients.
This is not merely a hypothetical concern. ERC-8004, a standard protocol for agent registration and reputation maintenance in permissionless AI markets, explicitly identifies Sybil manipulation as a security risk~\cite{erc8004}. A recent empirical study~\cite{xiong2026can} flags 59.2--90.6\% of reviewers across three ERC-8004 chains as participants in coordinated Sybil campaigns. Moreover, across these chains, the median cost of manipulating a reputation score past the study's trust threshold is only \$0.0027--\$0.055.

\para{This paper.}
We present \sysname (\Cref{sec-protocol}), a protocol for Sybil-resilient payment authorization in permissionless AI markets.
The core idea behind \sysname is to derive the reputation (or trustworthiness) of agent providers from a \emph{trust network}~\cite{kate2016introduction,moreno2018mind}.
Specifically, building on Bazaar's personalized-flow principle for limiting marketplace fraud~\cite{post2011bazaar}, \sysname models trust as a weighted, directed graph in which each edge has a finite capacity representing the amount of trust one agent extends to another. Payment authorization is then governed by these capacities: a payment is admitted only if the available capacities along paths from the sender to the receiver can collectively support the payment amount.

In such a trust network, every path from an honest sender to an adversarial receiver must cross an edge from the honest region to the adversarial region. These boundary edges constitute a directed \emph{cut}. At any time, \sysname bounds honest agents' aggregate payment exposure (i.e., the total value their clients risk losing to adversarial providers) by the total capacity of this cut, a property we call \emph{cut-bounded exposure} (\Cref{def:cut-bounded-exposure}).
This bound makes payment authorization \emph{Sybil-resilient}: creating identities or inflating capacities entirely within the adversarial region cannot enlarge the bound without additional capacity crossing the cut.

\para{Challenges and approaches.}
Achieving Sybil-resistant payment authorization and cut-bounded exposure in a permissionless AI market introduces three system challenges.

First, \emph{coordinating concurrent payment authorization} (\textbf{C1}). In the AI market, multiple concurrent payments may contend for the same edge capacity, with aggregate value exceeding what is available. Admitting them without accounting for capacity usage violates cut-bounded exposure, while permanently consuming capacity under a first-come-first-served policy prevents reuse after successful interactions and reduces admission opportunities.

\sysname addresses this challenge through \emph{atomic route reservation} and \emph{asynchronous feedback-based resolution}. Following the total order established by the blockchain, a smart contract atomically verifies that each payment route has sufficient available capacity, transfers the payment, and reserves the corresponding capacity on every edge along the route. Reserved capacity is temporarily unavailable to other payments and is later resolved asynchronously based on feedback from the sender. Positive or neutral feedback releases the reservation for future use, whereas negative feedback permanently consumes the reserved capacity. This mechanism preserves the cut-bounded exposure guarantee while allowing capacity associated with successful payments to be safely reused by subsequent payments.

Second, \emph{bounding capacity depletion by malicious senders} (\textbf{C2}). The same reservation discipline exposes a complementary attack. An adversary acting as a buyer can route payments through honest agents and deliberately issue negative feedback, consuming capacity that would otherwise support honest payments. 

\sysname combines an \emph{interaction-backed capacity discipline} with a \emph{bounded-stretch routing rule}. Capacity entering the honest region from an adversarial sender arises only after positive sender feedback on a payment to an honest receiver. Receiving a payment creates no capacity in the reverse direction. Consequently, an adversary must pay honest receivers to establish capacity entering the honest region, and negative feedback that depletes downstream honest capacity also consumes this ingress capacity. Moreover, the routing rule limits each path to $H_{\max}$ edges. A path from an adversarial sender that depletes honest capacity must first cross into the honest region. Each unit of consumed ingress can therefore deplete at most $H_{\max}-1$ units of capacity on edges between honest agents. We prove that cumulative depletion is at most $H_{\max}-1$ times the cumulative adversarial-to-honest payment value, ensuring a property called \emph{bounded depletion amplification} (\Cref{sec-property-goal}).

Third, \emph{verifying and resolving payments with bounded on-chain cost} (\textbf{C3}). In a trust network, agents need to find paths (called \emph {routes}) to support their payments. However, on-chain execution is replicated and metered, making graph-wide path discovery expensive. Moving discovery off-chain introduces untrusted and potentially stale routes. Moreover, the on-chain contract must verify the validity of the route and preserve necessary information for payment resolution. However, storing every payment's per-edge reservation records adds persistent storage costs that grow with route size.

To address \textbf{C3}, \sysname employs a \emph{compute--verify decomposition} and a \emph{witness-commit mechanism}. Specifically, \sysname allows independently operated pathfinders to discover routes off-chain, while the on-chain contract verifies the validity of submitted routes. By limiting path length and path count, \sysname bounds verification work independently of the trust network's size. 
The witness-commit mechanism further reduces both verification and storage costs by reusing an edge-load table across admission and resolution. Specifically, \sysname introduces a route witness that aggregates repeated edge loads, avoiding the repeated searches of an array-based verifier and supporting linear route verification. Moreover, instead of persisting the witness's per-edge reservation records, the contract stores only a constant-size commitment and authenticates the same load table during resolution. These techniques achieve linear verification and a constant persistent receipt state per admitted payment while preserving exact reservation accounting.

\para{Implementation and results.}
We implement a \sysname prototype and evaluate it using a graph derived from January--August 2026 Ethereum transactions and synthetic workloads (\Cref{sec-evaluation}). Our experimental results show that in \sysname, clients' exposure remains within the fixed cut bound under attacks with up to 10,000 Sybil identities, and constructed depletion workloads attain the amplification bound. Moreover, admitting a common payment costs approximately \$0.04--\$0.21 under the Ethereum price at the time of writing. In addition, across the tested routes, witness-commit reduces admission cost by up to 40.8\% compared with a baseline that reconstructs and stores per-edge reservation records.

\section{Background, Related Work, and Motivation} \label{sec-prelim}

\subsection{Blockchain Execution Model} \label{sec-bg-blockchain}
A blockchain implements a replicated state machine that orders transactions and deterministically executes them to update on-chain state~\cite{ethereum,solana,sui,aptos}. Users submit transactions that invoke \emph{smart contracts} to modify this state. Transaction execution is atomic: if a transaction succeeds, all of its state updates are committed; otherwise, none of them take effect.

Since transaction execution is replicated across validators, blockchains meter computational and storage resources, commonly through \emph{gas}~\cite{ethereumGas}, and impose execution limits to bound the work performed by validators. A transaction that exhausts its allocated resources fails, while block-level limits bound the aggregate execution work that can be included in a block. These limits make computationally expensive operations, e.g., graph-wide path discovery, unsuitable for on-chain execution.

A blockchain also serves as a public, verifiable ledger from which clients can obtain finalized state, transaction inputs, and execution logs. \sysname exploits these properties in three ways: it relies on atomic execution to coordinate concurrent transaction submissions, moves graph-wide path discovery outside replicated execution, and uses finalized on-chain data to recover unresolved receipt witnesses.


\subsection{Agent Registries and Reputation Risks} \label{sec-bg-8004}
ERC-8004~\cite{erc8004} is a proposed standard for permissionless AI-agent registries whose contracts have been live on Ethereum mainnet since January 2026~\cite{erc8004_mainnet}. It defines three registry interfaces. The \emph{Identity Registry} specifies how AI agents register on-chain identities and associated metadata. The \emph{Reputation Registry} allows agents to submit and retrieve feedback signals that can be used to assess the trustworthiness of other agents. The \emph{Validation Registry} enables agents to request independent verification of their work.

Recent measurements, however, reveal substantial risks in relying directly on these reputation signals. Xiong et al.~\cite{xiong2026can} find that feedback records in deployed ERC-8004 registries rarely include declared evidence of an underlying interaction. Across Ethereum~\cite{ethereum}, BNB Chain~\cite{bnbsmartchain}, and Base~\cite{base}, more than 98\% of feedback records contain neither a declared payment proof nor a task linkage. A shared-first-funder analysis flags 73.5\%, 59.2\%, and 90.6\% of reviewers on the three chains as participants in candidate coordinated campaigns. Removing feedback from those reviewers leaves 15.8\%, 77.9\%, and 86.8\% of rated agents without remaining feedback.
The same study further shows that reputation scores can be manipulated at very low cost. For rating-style tags and a representative trust threshold of 90, a single ceiling-valued feedback record moves the median rated agent above the threshold on each measured chain. The median gas cost of submitting such a record is only \$0.055 on Ethereum, \$0.0042 on BNB Chain, and \$0.0027 on Base. Consequently, an autonomous agent buyer that translates such a global reputation score directly into a payment decision may place an arbitrarily large value at risk, because the reputation mechanism itself imposes no protocol-level bound on the resulting financial exposure.


\begin{figure}[t]
    \centering
    \includegraphics[width=0.4\textwidth]{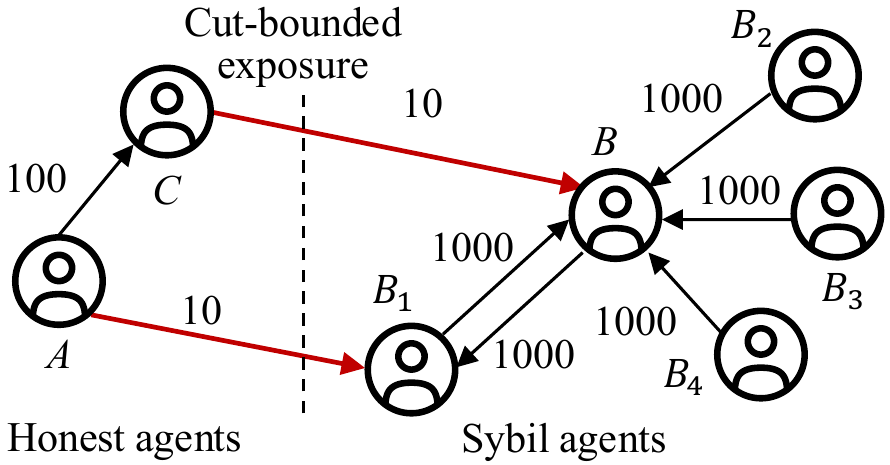}
    \caption{Bounding Sybil influence in trust networks. $A$ and $C$ are honest, while $B$ is malicious. Creating Sybil identities and inflating capacities within $B$'s region cannot increase the capacity that honest agents extend to that region.}
    \label{fig:bound-risk}
\end{figure}

\subsection{Trust Networks} \label{sec-bg-trust-network}
Trust networks~\cite{levien1998attack,walsh2006experience,kate2016introduction,moreno2018mind,maffei2017silentwhispers,viswanath2012canal,tran2009sybil} represent relationships as a weighted graph $\mathbb{G}=(\mathbb{V},\mathbb{E})$, where vertices represent participants and each edge carries a finite capacity extended by one participant toward another. Paths allow one participant to establish a personalized notion of trust in another through intermediaries. For a sender $s$ and receiver $d$, the aggregate feasible flow between them is bounded by the capacity of every cut separating $s$ from $d$. Consequently, adding identities, edges, or capacity entirely behind such a cut cannot increase this upper bound on feasible flow from $s$ to $d$ without increasing the capacity crossing the cut.

\Cref{fig:bound-risk} illustrates how this property limits the influence of Sybil identities.
Agents $A$ and $C$ are honest, while adversary $B$ controls Sybil identities $B_1,\ldots,B_4$. Every path from the honest region into the adversarial region must cross one of the two red edges, each with capacity 10. These edges form a cut of capacity 20, which bounds the aggregate flow from the honest region into the adversarial region. The adversary may create additional identities or arbitrarily increase capacities among identities within its own region (e.g., to 1,000), but doing so cannot increase this bound unless additional capacity is established across the cut.

Prior systems use this capacity constraint to bound resource use in several applications. Ostra~\cite{mislove2008ostra} uses credit to limit unwanted communication, and SumUp~\cite{tran2009sybil} uses adaptive flow over a trust graph to bound the influence of votes. Bazaar~\cite{post2011bazaar} applies personalized maximum flow to bound transaction risk in centralized online marketplaces. It temporarily reserves path capacity and consumes them when fraud is reported. These systems exemplify \emph{Sybil tolerance} through bounded resource use. Canal~\cite{viswanath2012canal} makes this approach scalable by approximating credit payments with landmark routing.

Trust networks also serve as payment infrastructure. Ripple~\cite{ripple} uses credit limits to support value transfer, and Dandekar et al.~\cite{dandekar2011liquidity} analyze how topology and credit capacity affect liquidity under repeated payments. PrivPay~\cite{pedro2015privacy} and SilentWhispers~\cite{maffei2017silentwhispers} provide privacy-preserving credit-network payments, while SpeedyMurmurs~\cite{roos2017settling} develops efficient, privacy-preserving routing for path-based transactions. On-chain systems demonstrate how routing can be separated from execution: Circles~\cite{circles} routes personal tokens through trust connections, while Trustlines~\cite{trustlines} implements transfers through bilateral credit relationships. Both use off-chain routing services to prepare transfers executed by smart contracts.

These works motivate \sysname's use of personalized flow to address the unbounded financial exposure from registry reputation signals in permissionless AI markets (\Cref{sec-bg-8004}). \sysname builds on Bazaar's reservation and feedback semantics and the separation of off-chain routing from on-chain execution exemplified by Circles and Trustlines. We present \sysname in~\Cref{sec-protocol}.

\section{System Model and Design Goals} \label{sec-problem-definition}

We consider a permissionless AI market in which agents discover services, select providers, and exchange value without relying on a central operator. Payment authorization is governed by personalized capacity in the trust network: an admitted payment atomically transfers value and reserves the capacity supporting its authorization. 
Feedback subsequently resolves the reservation and updates the trust network, releasing or consuming reserved capacity according to the sender's assessment of the service outcome.
\Cref{sec-system-model} and \Cref{sec-adversary-model} formalize the system and adversary models, respectively, while \Cref{sec-property-goal} defines the security properties and design goals of \sysname.

\subsection{System Model} \label{sec-system-model}
\para{Entities and system parameters.}
The system consists of clients, AI agents, off-chain pathfinders, a blockchain, and a \scname smart contract. A client may create and control one or more agents, and each registered agent has a stable identifier.
Following systems that bootstrap trust from pre-trusted nodes or trusted seeds~\cite{gyongyi2004trustrank,kamvar2003eigentrust,cao2012aiding,alvisi2013sok}, we assume a set of honest \emph{TrustRoots} $\mathbb T$. During system initialization, \sysname assigns finite available capacity to selected outgoing edges of these TrustRoots, thereby seeding the initial trust network.
A pathfinder reconstructs a finalized view of the trust graph from public blockchain data and computes candidate routes for requested payments. Before submitting a payment, the requesting client obtains a corresponding route from a pathfinder. Any party may run a pathfinder locally or provide pathfinding as a service. 
The outputs of pathfinders are untrusted, and the \scname contract independently validates all route-dependent state transitions.
We treat the blockchain as a black box throughout this paper and model it as a state machine replication. The \scname contract maintains the authoritative trust-network state and validates every state-changing operation. We use logical time $t$ to denote the system state immediately after a finalized contract transition.

We assume that all payment values and trust-network capacities are denominated in a single fixed payment asset and that each transfer delivers its stated amount. A deployment may use a stablecoin~\cite{usdc,usdt,dai,pyusd} to reduce value volatility. This common denomination enables direct capacity accounting without a price oracle.

Moreover, \sysname introduces several system parameters to bound cost and adversarial impact. Parameter $H_{\max}\geq 2$ bounds the maximum path length in edges, $K_{\max}$ bounds the maximum number of paths per route, and $U_{\max}$ bounds the maximum number of unique edges across a route.

\para{Trust-network state.}
At time $t$, the trust network $\network$ is a weighted, directed graph $\mathbb{G}_t=(\mathbb{V}_t,\mathbb{E}_t)$. Each vertex $s\in\mathbb{V}_t$ represents an agent. Each directed edge $e_{s,d}\in\mathbb{E}_t$ represents trust extended by $s$ toward $d$. Each edge maintains \emph{available capacity} $a_t(s,d)\geq 0$, which may support new transactions, and \emph{reserved capacity} $r_t(s,d)\geq 0$, which is committed to admitted but unresolved transactions. Its total capacity is $c_t(s,d)=a_t(s,d)+r_t(s,d)$.

\para{Routed transactions.}
A transaction of value $W>0$ from sender $s$ to receiver $d$ is supported by a route
\tightdisplay{
    \mathcal{P}=\{(P_1,w_1),\ldots,(P_k,w_k)\},
}%
where $1\leq k\leq K_{\max}$, each $P_i$ is a simple directed path from $s$ to $d$ with $1\leq |P_i|\leq H_{\max}$, and $w_i>0$ is the amount routed along $P_i$, with $\sum_i w_i=W$. We write $L(\mathcal{P})=\sum_i |P_i|$ for the total number of edge occurrences in the route, and thus $L(\mathcal{P})\leq K_{\max}H_{\max}$. The aggregate load placed by $\mathcal{P}$ on an edge $e$ is
\tightdisplay{
\ell_e(\mathcal{P}) =
\sum_{i\in P_i} w_i.
}%
A route is \emph{feasible} at $t$ if $\ell_e(\mathcal{P})\leq a_t(e)$ for every edge it uses.

\para{Route witness.}
Let $e_1\prec\cdots\prec e_U$ denote the distinct edges with positive aggregate used by $\mathcal{P}$, ordered lexicographically by their endpoints. The \emph{route witness} $\Lambda(\mathcal{P})=((e_1,\ell_{e_1}(\mathcal{P})),\ldots,(e_U,\ell_{e_U}(\mathcal{P})))$ records the aggregate load on each distinct route edge. It is supplied with the route at admission and later re-supplied whenever the routed reservation is resolved (\Cref{sec-route-witness-commit}).


\para{Trust-network interface.}
The trust network exposes the following abstract operations. 
Let $\mathcal{F}=\{\mathsf{pos},\mathsf{neu},\mathsf{neg}\}$ denote the feedback domain and let $\mathcal{R}=\{\mathsf{snd},\mathsf{rcv}\}$ denote the participant role (sender or receiver) being resolved.
\begin{itemize}
    \item $\network.\text{setup}() \rightarrow \mathsf{pp}$ initializes the trust network and publishes the system parameters $\mathsf{pp}$.
    \item $\network.\text{register}(\mu) \rightarrow s$ binds agent metadata $\mu$ to the caller, assigns the agent a unique identifier $s$, and records the caller as its owner.
    \item $\network.\text{joinProvider}(d,r,W) \rightarrow \mathsf{ok}$ allows a new agent $d$ (who wants to serve as a provider later) to obtain an initial incoming edge from TrustRoot $r$ by paying a fee $W$ (\Cref{sec-protocol-bootstrap}).
    \item $\network.\text{seedPay}(s,r,d,W) \rightarrow \rho$ allows a newly registered agent $s$ to pay $W$ to provider $d$ through TrustRoot $r$ (\Cref{sec-protocol-bootstrap}). It returns a receipt $\rho$ binding the transaction to that reservation.
    \item $\network.\text{transact}(s,d,W,\mathcal{P},\Lambda) \rightarrow \rho$ admits a transaction of value $W$ from sender $s$ to receiver $d$ with route witness $\Lambda$, atomically reserves its per-edge loads, and returns a receipt $\rho$.
    \item $\network.\text{feedback}(\rho,\eta,f,\widehat{\Lambda}) \rightarrow \mathsf{ok}$ accepts feedback $f\in\mathcal{F}$ for role $\eta\in\mathcal{R}$, updates trust, and resolves any reservation.
    \item $\network.\text{resolve}(\rho,\eta,\widehat{\Lambda}) \rightarrow \mathsf{ok}$ allows any client to resolve role $\eta\in\mathcal{R}$ if its feedback is still missing after a timeout.
\end{itemize}

\subsection{Threat Model and Assumptions} \label{sec-adversary-model}

We partition the agents into an honest set $\mathbb{H}$ and an adversarial set $\mathbb{A}$, where the TrustRoots $\mathbb T\subseteq\mathbb H$.
Honest agents follow the protocol, while the adversary can behave arbitrarily, such as creating arbitrarily many Sybil identities, coordinating all agents in $\mathbb{A}$, and constructing arbitrary internal trust relationships. The adversary may also interact freely with honest agents subject to the capacities that honest agents extend toward it.
However, the adversary cannot compromise the blockchain's security guarantees. 
For honest clients, we assume they can resolve their receipts before timeout. We further assume that, for every unresolved receipt whose timeout has expired, an honest client eventually invokes timeout resolution and can retrieve the receipt's finalized admission data or load events from the blockchain. Such resolution can be incentivized in practice, but designing the corresponding incentive mechanism is outside the scope of this paper.

\subsection{Security and Efficiency Goals} \label{sec-property-goal}

This section states the security properties and efficiency goals that \sysname is designed to satisfy. 
We first introduce notation for directed cuts and their capacities. For disjoint agent sets $\mathbb{X},\mathbb{Y}\subseteq\mathbb{V}_t$, define the directed cut from $\mathbb{X}$ to $\mathbb{Y}$ as
\tightdisplay{
    \delta_t(\mathbb{X},\mathbb{Y})=\{e_{s,d}\in\mathbb{E}_t:s\in\mathbb{X},\,d\in\mathbb{Y}\}
}%
with total capacity
\tightdisplay{
    C_t(\mathbb{X},\mathbb{Y})=\sum_{e\in\delta_t(\mathbb{X},\mathbb{Y})}c_t(e).
}

\begin{definition}[\textbf{Cut-Bounded Exposure}] \label{def:cut-bounded-exposure}
    Let $F_{\mathbb H}(t)$ denote the maximum aggregate payment loss that honest agents may incur at time $t$. A system provides cut-bounded exposure if
    \tightdisplay{
        F_{\mathbb H}(t)\leq C_t(\mathbb H,\mathbb A)
    }%
    for every partition $(\mathbb H,\mathbb A)$ and every time $t$.
\end{definition}

\begin{definition}[\textbf{Bounded Depletion Amplification}] \label{def:bounded-depletion}
    Let $D_{\mathbb H}(t)$ denote the cumulative capacity consumed on honest-to-honest edges by negative sender feedback for payments initiated by agents in $\mathbb A$, and let $P_{\mathbb A\rightarrow\mathbb H}(t)$ denote the cumulative value transferred from adversarial agents to honest agents up to time $t$. A system provides $\alpha$-bounded depletion amplification if
        \tightdisplay{
        D_{\mathbb H}(t) \leq \alpha P_{\mathbb A\rightarrow\mathbb H}(t)
        }%
for every partition $(\mathbb H,\mathbb A)$ and every time $t$.
\end{definition}

Definitions~\ref{def:cut-bounded-exposure} and~\ref{def:bounded-depletion} capture complementary adversarial risks. Cut-bounded exposure gives \sysname Sybil-resilient payment authorization. In particular, the bound on honest agents' payment exposure applies collectively to all adversarial identities and cannot be enlarged by adding identities, exchanging collusive feedback, or circulating payments entirely within the adversarial region (\Cref{cor:sybil-invariance}). 
Bounded depletion amplification limits the honest-to-honest capacity that malicious buyers can permanently consume relative to their cumulative payments to honest agents.

\begin{definition}[\textbf{Route Soundness}] \label{def:route-soundness}
    A system provides route soundness if every admitted routed transaction is supported by a route that is feasible in the contract state at admission, and the protocol reserves and later resolves exactly the route's aggregate edge loads.
\end{definition}

\begin{definition}[\textbf{Eventual Resolution}] \label{def:eventual-resolution}
    A system provides eventual resolution if every admitted transaction can eventually be resolved, after which no capacity associated with the transaction remains reserved.
\end{definition}

\para{Efficiency goal: linear verification and constant persistent state.}
\sysname should avoid graph-wide on-chain computation because blockchain execution is constrained by gas limits (\Cref{sec-bg-blockchain}). 
The cost of admitting and resolving a transaction should be independent of the overall trust-network size and scale at most \emph{linearly} with the submitted route, rather than quadratically with its size. Moreover, each admitted transaction should require only \emph{constant} persistent contract state, rather than storing state proportional to the route.
Our goal is to construct a permissionless payment-authorization protocol that simultaneously provides cut-bounded exposure, bounded depletion amplification, route soundness, and eventual resolution while satisfying this efficiency goal.


\section{\sysname Design} \label{sec-protocol}
\para{Overview.}
\sysname is a payment-authorization protocol for permissionless AI markets. To initiate a payment of value $W$, a sender typically obtains a personalized route from an off-chain pathfinder. The on-chain \scname contract verifies that the submitted route is feasible, transfers the payment, and atomically reserves its aggregate edge loads. Participant feedback subsequently resolves the reservation and updates the trust network. \Cref{fig:workflow} illustrates this workflow. \sysname operates in four phases over the trust-network interface defined in \Cref{sec-system-model}.
Phase 0 initializes the \scname contract and system parameters. Phase 1 allows clients to register agents and bootstrap initial trust edges. Phase 2 admits routed payments subject to available edge capacity and atomically reserves their supporting capacity. Phase 3 resolves these reservations through participant feedback or permissionless timeout resolution and updates the trust network accordingly.
We describe these phases in this section and defer their security and complexity analysis to \Cref{sec-analysis}.

\sysname addresses the three challenges identified in \Cref{sec-intro} while satisfying the security properties and efficiency goals defined in \Cref{sec-property-goal}.
For \textbf{C1}, \emph{atomic route reservation} ensures that concurrent payments cannot reuse the same available capacity: once a payment is admitted, its aggregate edge loads become reserved and remain unavailable until resolution. Since every route from an honest sender to an adversarial receiver must cross the honest-to-adversarial cut, this accounting bounds honest agents' aggregate payment exposure by the cut capacity, providing cut-bounded exposure (\Cref{def:cut-bounded-exposure}). Feedback-based resolution subsequently releases capacity after positive or neutral feedback, allowing capacity from successful interactions to be reused rather than permanently consumed.

For \textbf{C2}, \sysname combines an \emph{interaction-backed capacity discipline} with a \emph{bounded-stretch routing rule}. The former requires adversarial-to-honest ingress capacity to arise from paid interactions, while the latter limits every route path to at most $H_{\max}$ edges (\Cref{sec-protocol-transaction}). These mechanisms bound the honest-to-honest capacity that malicious senders can consume through negative feedback relative to the value they transfer to honest agents, yielding $(H_{\max}-1)$-bounded depletion amplification (\Cref{def:bounded-depletion}).

For \textbf{C3}, \sysname uses a \emph{compute--verify decomposition} (\Cref{sec-route-decomposition}) that moves graph-wide route discovery to untrusted off-chain pathfinders while retaining deterministic on-chain verification, thereby providing route soundness (\Cref{def:route-soundness}). To avoid storing per-edge reservation records for every admitted payment, the \emph{witness-commit mechanism} (\Cref{sec-route-witness-commit}) stores only a constant-size commitment while preserving a recoverable route witness for later resolution. This design achieves linear route verification and a constant persistent contract state. Finally, permissionless timeout resolution (\Cref{sec-protocol-feedback}) resolves reservations when participant feedback is missing, providing eventual resolution (\Cref{def:eventual-resolution}).

\begin{figure*}
    \centering
    \includegraphics[width=\textwidth]{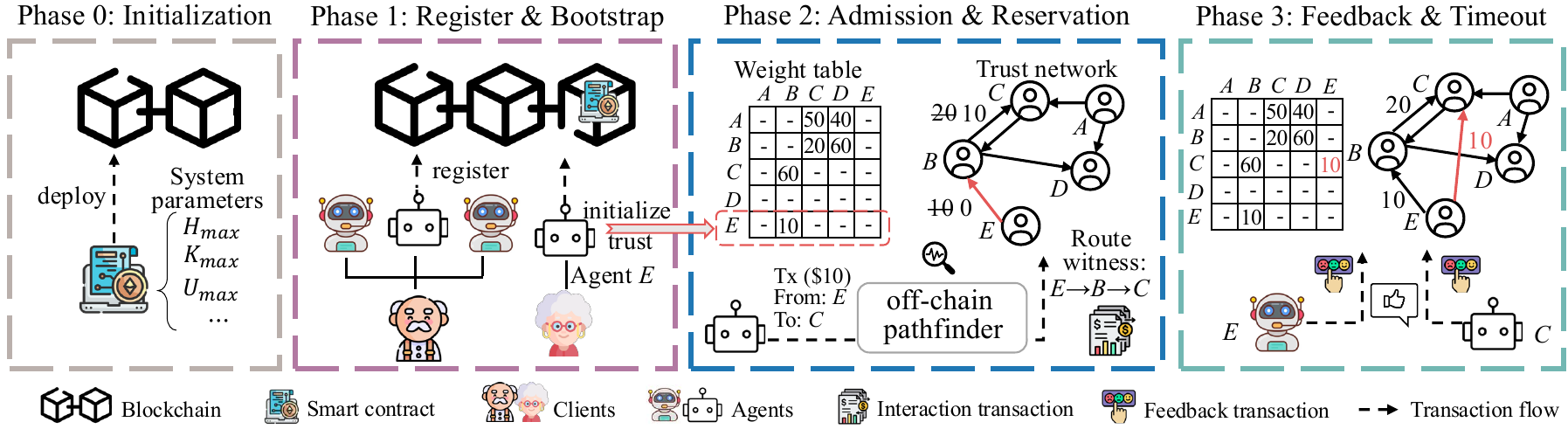}
    \caption{The workflow of \sysname. Phase 0: the \scname contract is deployed and initialized on-chain. Phase 1: clients register agents and bootstrap initial trust edges. Phase 2: agents interact through pathfinders off-chain and are verified on-chain. Phase 3: participant feedback updates the trust graph and resolves reserved capacity.}
    \label{fig:workflow}
\end{figure*}

\subsection{Protocol Initialization} \label{sec-protocol-setup}
The protocol begins by deploying and initializing the \scname contract through $\network.\text{setup}()$. Initialization fixes the contract address $\mathit{Addr}{\scname}$, feedback timeout $T_{\mathrm{fb}}$, maximum path length $H_{\max}$, maximum number of paths per route $K_{\max}$, and maximum number of unique route edges $U_{\max}$. We collect these values into the public parameters $\mathsf{pp}$.
The protocol chooses $H_{\max}$, $K_{\max}$, and $U_{\max}$ such that verifying up to $K_{\max}H_{\max}$ edge occurrences and $U_{\max}$ distinct edge records fits within the designated on-chain resource budget. 
During initialization, \sysname also seeds the trust network by assigning finite available capacity to a configured set of edges originating from the TrustRoots, with all reserved capacities initially set to zero. These seeded edges provide the initial trust from which subsequent payment routes can develop.

Initialization occurs once. Afterward, any client may invoke the public \scname interface without authorization from a marketplace operator.

\subsection{Agent Registration and Trust Bootstrap}
\label{sec-protocol-bootstrap}
After setup, clients may register agents and bootstrap trust through TrustRoot-supported capacity grants and seed payments. Registration creates a stable on-chain identity. A new provider agent can obtain an initial incoming trust edge, while a new buyer agent can make its first payment without already having a trust edge to the TrustRoot or the selected provider.
\para{Agent registration.}
A client invokes $\network.\text{register}(\mu)$ to register an agent on-chain, where $\mu$ contains the agent's metadata, such as its name, description, and offered services.
The \scname contract assigns the agent a unique identifier $s$ and records the association between $s$ and the client's on-chain address. The client thereafter controls the agent and uses $s$ as its stable identity in the network. A client may register multiple agents, but each agent receives a distinct identifier and maintains an independent trust profile as a separate vertex in the network.

\para{Trust bootstrapping.}
Every newly registered agent can attempt to bootstrap its trust as a service provider or as a buyer through a TrustRoot.
A TrustRoot may preauthorize a finite capacity budget for automatic provider bootstrapping at setup. A new service provider $d$ invokes $\network.\text{joinProvider}(d,r,W)$ with a TrustRoot $r\in\mathbb T$ and pays an irrecoverable fee $W>0$. If the root's budget permits the grant, \scname atomically charges the fee and adds $W$ available capacity to $e_{r,d}$ without requiring the root to approve $d$ individually. 
This mechanism gives trust bootstrapping a practical interpretation as a market entry cost. A provider can pay a modest fee to support small initial purchases, then earn additional incoming capacity through positive buyer feedback on paid interactions (\Cref{sec-protocol-feedback}). The fee also ties bootstrap exposure to an explicit economic cost: because TrustRoots are honest, granting $W$ capacity to an adversarial provider increases the honest-to-adversarial cut capacity by $W$ and requires an irrecoverable payment of the same amount. Creating additional Sybil identities therefore cannot obtain bootstrap capacity for free.

To bootstrap its first outgoing trust relationship (as a buyer), a newly registered agent $s$ invokes $\network.\text{seedPay}(s,r,d,W)$, specifying a TrustRoot $r\in\mathbb T$, a service provider $d$, and a payment value $W>0$. It is treated as a special routed payment over the logical path $s\rightarrow r\rightarrow d$: no trust edge or capacity is required for the first hop, while the direct edge $e_{r,d}$ must have at least $W$ available capacity. If admitted, the protocol transfers $W$ directly from $s$ to $d$ and reserves $W$ capacity on $e_{r,d}$. This reservation is subsequently resolved using the same feedback interfaces and timeout rules as ordinary routed payments (\Cref{sec-protocol-transaction} and \Cref{sec-protocol-feedback}).


\subsection{Payment Admission and Capacity Reservation}
\label{sec-protocol-transaction}
After bootstrapping, an owner-controlled client for agent $s$ may select service provider $d$ and initiate a payment of value $W$. \sysname admits the payment only if the client supplies a feasible route $\mathcal{P}$ from $s$ to $d$ that carries total value $W$. 
The client therefore obtains a candidate route from an off-chain pathfinder before submitting the payment. Successful admission atomically transfers $W$, reserves the route's aggregate edge loads, and creates a receipt that binds subsequent resolution to this interaction. If admission fails, the transaction reverts without modifying the \scname state.

Admission separates three responsibilities: off-chain route discovery, on-chain route verification and atomic capacity reservation, and receipt creation for subsequent resolution. To perform these tasks efficiently, \sysname combines a \emph{compute--verify decomposition} with a \emph{witness-commit mechanism}. 
The former keeps graph-wide path discovery outside replicated execution. The latter verifies a route witness at admission and commits to it for reuse during resolution, enabling linear verification and constant persistent receipt state.

\subsubsection{Compute--verify Decomposition} \label{sec-route-decomposition}
\sysname performs route discovery off-chain and verifies the resulting route on-chain. Because the trust-network state is publicly reconstructible from finalized blockchain data, a client may maintain the graph locally and run any compatible pathfinding algorithm. Resource-constrained clients may instead query independently operated pathfinders that index the \scname state. Pathfinder outputs are untrusted: a stale, malformed, or otherwise infeasible route is rejected by the contract without modifying the trust-network state.

Our prototype pathfinder adapts the multipath routing algorithm of Bazaar~\cite{post2011bazaar} with explicit search bounds. It stops extending a path once the path reaches $H_{\max}$ edges and terminates once the selected paths collectively carry $W$ or their number reaches $K_{\max}$. If it cannot obtain total value $W$ within these bounds, the pathfinder reports failure.

\subsubsection{The Witness-commit Mechanism} \label{sec-route-witness-commit}
Moving route discovery off-chain avoids graph-wide on-chain search, but two costs remain: the contract must verify aggregate loads when different paths share edges, and it must retain enough authenticated information to resolve the resulting reservations later. \sysname addresses both costs with the \emph{witness-commit mechanism}, which reuses the same edge-load table across admission and resolution. At admission, the contract verifies the table against the submitted route and uses its loads to reserve capacity. A constant-size commitment authenticates the same table when it is later re-supplied for resolution, replacing route-sized persistent receipt state.

\para{Route witness.}
To admit a payment, the contract must sum the loads contributed by all paths that share an edge and check that the total does not exceed that edge's available capacity. A simple implementation processes the $L=L(\mathcal{P})$ edge occurrences by repeatedly searching an array of up to $U$ distinct edges to accumulate the load on each unique edge, requiring $O(LU)$ work. EIP-1153~\cite{eip1153} supports linear-time aggregation through transient storage. However, transient reads and writes cost more gas than ordinary memory accesses, and the accumulated loads must be cleared or kept separate across multiple contract calls within the same transaction. \sysname instead has the client prepare a sorted edge-load table and supply the index of the table entry corresponding to each edge occurrence. The contract checks these indices and recomputes the loads in a temporary memory array to verify the supplied totals. This enables linear verification using memory that is local to the current contract call, without managing temporary totals across calls. The verified table also serves as the route witness committed at admission and re-supplied during resolution, linking both phases through the same representation.

Specifically, the client constructs the route witness $\Lambda(\mathcal{P})$ defined in \Cref{sec-system-model}: a lexicographically sorted, duplicate-free table of the $U$ distinct route edges and their aggregate loads. The submission associates each path-edge occurrence with its corresponding witness entry and provides a sorted node-order certificate for each path. During a single scan of the route, \scname verifies these associations and certificates, accumulates each occurrence directly into its indexed witness entry, and checks the resulting loads against the claimed witness table. This procedure verifies path simplicity and aggregate loads in $O(L+U)$ work.

\para{On-chain route verification and reservation.}
The sender $s$ invokes $\network.\text{transact}(s,d,W,\mathcal{P},\Lambda(\mathcal{P}))$ to submit a transaction $Tx_i$ with the discovered route and witness. 
The encoding exposes the number of paths and each path length before its vertex sequence, allowing \scname to enforce $1\leq k\leq K_{\max}$ and $|P_i|\leq H_{\max}$ before processing each path.
Subject to these size checks, \scname verifies that the caller owns $s$ and has sufficient funds, every path starts at $s$ and ends at $d$, consecutive vertices form existing directed edges, every path is simple, every $w_i>0$, $\sum_iw_i=W$, and the witness contains exactly the aggregate route loads. It also checks that the witness contains at most $U_{\max}$ unique edges. For every unique edge $e$, it then checks $\ell_e(\mathcal{P})\leq a_t(e)$ against the contract state immediately before admission.

If any check fails, the transaction reverts without changing the state. Otherwise, \scname transfers $W$ from $s$ to $d$ and atomically moves $\ell_e(\mathcal{P})$ units from available to reserved capacity on every involved edge. Since the checks and updates are atomic, concurrently submitted transactions ordered by the blockchain will not consume the same reservation.

\para{Constant-state receipt commitment.}
Later resolution must recover exactly the edge loads reserved by an admitted transaction. A naive design could persist all $U$ edge-load records from $\Lambda(\mathcal{P})$ in contract storage until resolution, incurring $O(U)$ persistent state and writes per admitted transaction.

Instead, after validating the route witness $\Lambda_i=\Lambda(\mathcal{P})$, \scname stores constant-size receipt metadata together with a receipt-bound commitment $C_{\rho_i}=\mathsf{Hash}(\Lambda_i)$.
When the receipt is later resolved, the resolver re-supplies the same load table $\Lambda_i$ verified at admission as transaction input. The contract verifies its hash against $C_{\rho_i}$ before applying any capacity update, ensuring that resolution operates on exactly the edge loads reserved at admission. Resolution requires only this load table, without resupplying the paths or their verification certificates. Because the witness is included in the finalized admission transaction, it can also be recovered from public blockchain data when timeout resolution is required. Consequently, \sysname reduces persistent receipt state from $O(U)$ to $O(1)$ without sacrificing exact reservation resolution.

\subsection{Feedback and Timeout Resolution}
\label{sec-protocol-feedback}
Feedback resolves admitted transactions and updates trust. After a transaction $Tx_i$ creates receipt $\rho_i$, either participant may invoke $\network.\text{feedback}(\rho_i,\eta,f,\widehat{\Lambda}_i)$, where $\eta\in\{\mathsf{snd},\mathsf{rcv}\}$ denotes the sender or receiver and $f\in\{\mathsf{pos},\mathsf{neu},\mathsf{neg}\}$.
The two roles are finalized independently because they express directional trust judgments. Sender feedback resolves the payment reservation and may update trust from $s$ to $d$, whereas receiver feedback only updates trust from $d$ to $s$.
\Cref{fig:feedback-illustration} illustrates an interaction in which sender $E$ pays receiver $C$ value $W=5$ through $E\rightarrow B\rightarrow C$. Immediately after admission, five units on both route edges are reserved.

\begin{figure*}[t]
    \centering
     \includegraphics[width=0.9\textwidth]{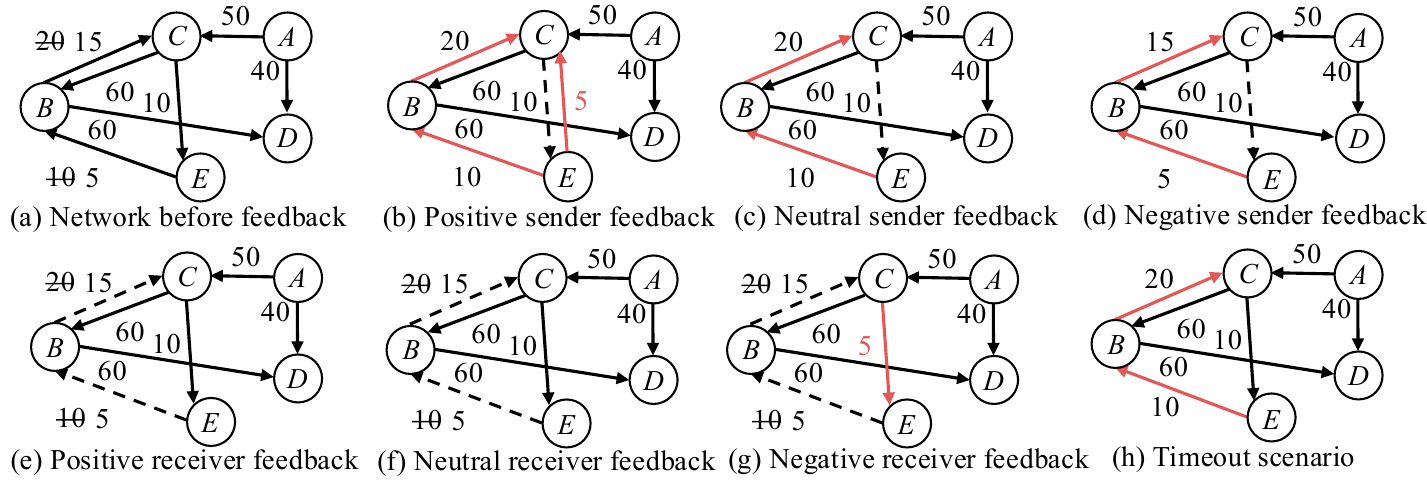}
    \caption{Feedback for a payment from $E$ to $C$ over $E\rightarrow B\rightarrow C$. 
    Red arrows indicate updated edges and dashed arrows unaffected reverse edges.
    (a) Admission reserves route capacity. (b)--(d) Sender feedback. (e)--(g) Receiver feedback. (h) Timeout resolution.}
    \label{fig:feedback-illustration}
\end{figure*}

\para{Sender feedback.}
Positive feedback releases the reserved route loads back to available capacity and creates or increases the direct edge $e_{s,d}$ by $W$, since the sender is satisfied with the service and wants to build stronger trust toward the receiver (\Cref{fig:feedback-illustration}(b)). This capacity is \emph{interaction-backed} since it arises from an executed payment. Neutral feedback releases the reservation without changing $e_{s,d}$ (\Cref{fig:feedback-illustration}(c)). Negative feedback permanently consumes the reserved loads, reducing the total capacity of every involved edge accordingly, and creates no new capacity (\Cref{fig:feedback-illustration}(d)). The same rules apply to seed payments, whose only reserved load is $W$ on $e_{r,d}$.




\para{Receiver feedback.}
Receiver feedback neither resolves the route reservation nor creates capacity. Positive and neutral feedback leaves the trust graph unchanged (\Cref{fig:feedback-illustration}(e)--(f)). Negative feedback reduces the available reverse capacity by $\min{W,a_t(d,s)}$, and hence total capacity by the same amount, while leaving reserved capacity unchanged (\Cref{fig:feedback-illustration}(g)). If $e_{d,s}$ does not exist, no graph update occurs.





\para{Feedback timeout.}
Each feedback role can be applied at most once. If a role remains unresolved after $T_{\mathrm{fb}}$, any client may invoke $\network.\text{resolve}(\rho_i,\eta,\widehat{\Lambda}_i)$. For a missing sender role, the resolver reconstructs $\Lambda_i$ from finalized admission data and applies the neutral transition, releasing the reservation without creating capacity. A missing receiver role is resolved as a no-op. The contract then marks the corresponding role as final. Under the timeout-resolution assumption of \Cref{sec-adversary-model}, no admitted transaction retains reserved capacity indefinitely.

\section{Implementation and Evaluation} \label{sec-evaluation}

Our evaluation aims to answer the following questions: 
\begin{itemize}
    \item \textbf{RQ1--Bound behavior}: Does \sysname preserve cut-bounded exposure, and how tight is its depletion-amplification bound (\Cref{sec-eval-sybil-resist})?
    \item \textbf{RQ2--On-chain efficiency}: What is \sysname's on-chain execution cost, and how do route complexity and receipt mechanisms affect it (\Cref{sec-eval-on-chain-cost})?
    \item \textbf{RQ3--Routing utility and scalability}: How do route bounds, network scale, and pathfinding algorithms affect routing utility and discovery latency (\Cref{sec-eval-route})?
\end{itemize}
We also evaluate topology-sensitive liquidity loss and transaction contention. We defer these results to \Cref{sec-more-experiments}.

\para{Implementation and setup.}
We implement a prototype for \sysname that comprises three main components: a \scname contract in Solidity~\cite{solidity}, an off-chain pathfinder, and clients in Python~\cite{python}. The \scname contract implements all the on-chain interfaces listed in \Cref{sec-system-model}. 
The off-chain pathfinder runs a contract indexer to collect state transitions for the \scname contract through emitted events and reconstruct the updated trust network. It runs a path discovery algorithm to find feasible routes for payments submitted by clients.
A client can register agents, bootstrap their first trust edges, query the pathfinder for feasible routes, construct route witnesses, submit payments, recover route witnesses, and drive feedback and timeout resolution. 
The total implementation includes around 4.4K lines of code (LOC) for the prototype and around 3.2K LOC for tests.

We run our evaluation on a Linux Ubuntu server with 48 CPU cores and 128 GB of RAM. 
We compile the \scname contract to Ethereum Virtual Machine (EVM) bytecode using the Foundry~\cite{foundry} toolkit and deploy it to an Anvil-based Ethereum chain. Anvil~\cite{anvil} provides a local evaluation environment. It executes contract transactions using EVM gas accounting and mines each submitted transaction locally, allowing us to control the initial state and repeat each experiment on a fresh deployment.

\para{Baselines.}
Since there is no single protocol implementing all functionalities studied by \sysname, we use metric-specific comparison mechanisms.

For Sybil exposure (\Cref{sec-eval-sybil-resist}), (1) \emph{global-count} is a policy that aggregates positive feedback records, and (2) \emph{global-volume} aggregates reported positive-interaction values. These are two unrestricted global policies that the adversary can exploit to inflate their reputation with Sybil attacks, as discussed in \Cref{sec-intro}. In contrast, \sysname adopts a personalized-flow principle built on a trust network to prevent the Sybil risk. We calibrate each policy's linear score-to-value conversion so that its one-identity exposure equals \sysname's cut capacity. 

For on-chain efficiency (\Cref{sec-eval-on-chain-cost}), (1) the \emph{route-aggregation store} mechanism accepts the paths rather than a route witness, reconstructs aggregate loads in $O(LU)$ time, and stores all $U$ edge-load records with the receipt, and (2) the \emph{witness-store} mechanism verifies a route witness in $O(L+U)$ time and stores the same $U$ records.
\sysname's witness-commit mechanism also verifies the route witness in $O(L+U)$ time, but replaces the stored edge-load table with one commitment and therefore requires only $O(1)$ persistent receipt storage.

For routing (\Cref{sec-eval-route}), \sysname uses a \emph{capacity-scaling BFS multipath} algorithm motivated by \cite{post2011bazaar}. It initializes a residual graph and processes power-of-two thresholds $\gamma$ from the largest no greater than the maximum available capacity down to one. At each threshold, it runs a hop-bounded breadth-first search for a simple $s$--$d$ path containing only residual edges of capacity at least $\gamma$, augments by the path bottleneck or the remaining transaction value, and subtracts that load from the residual graph. It stops when the selected paths carry the transaction value $W$, when it has selected $K_{\max}$ paths, or when all thresholds are exhausted. For comparison, we implement (1) the \emph{unscaled BFS multipath algorithm} that repeatedly augments over all positive residual edges, and (2) the \emph{repeated widest-path algorithm} that repeatedly selects the hop-bounded path with maximum bottleneck capacity. 

\para{Dataset.}
We evaluate \sysname with two datasets: the real-world \emph{Ethereum-derived} dataset and the \emph{synthetic} dataset. First, we collect transaction data from Ethereum mainnet blocks from January 29, 2026, when ERC-8004 was deployed, through August 2, 2026. We use the collected data to construct a trust graph and set up the initial states of \scname by applying several filter rules (\Cref{sec-dataset-detail}). The final Ethereum-derived graph contains 1,055 vertices and 1,628 directed edges with the median capacity of 0.00494~ETH. Since the permissionless AI agent market is in its infancy, this graph is sparse and small. We therefore also generate synthetic strongly connected directed graphs of 128, 512, 1,024, 2,048, 4,096, and 8,192 vertices with six directed edges per vertex and capacities sampled uniformly from $[10,100]$ to separate system scale from the sparsity of the emerging deployment.

\subsection{Sybil Exposure and Capacity Depletion} \label{sec-eval-sybil-resist}

We first evaluate \sysname's bound behavior using both Ethereum-derived and synthetic datasets. 

\para{Cut-bounded exposure} (\Cref{fig:eval-sybil}). 
In the Ethereum-derived setting (\Cref{fig:eval-sybil}(a)), we select five disjoint regions from the constructed graph, reflected by the shadow with the upper boundary as the maximum incoming capacity to the region and the lower boundary as the minimum incoming capacity. The selected regions are marked as Sybil regions, and each region's incoming capacity is the honest-to-adversarial cut capacity $C_t(\mathbb H,\mathbb A)$. For each region, we denote this fixed snapshot value by $C$ below. We then append $m\in\{1,10,100,1000,10000\}$ synthetic Sybil vertices in the region, each of which has an edge with capacity $C$ into the region. The synthetic setting is designed for graph structure generalization (\Cref{fig:eval-sybil}(b)), in which we construct a gateway graph for each Sybil count $m$. We set the honest-to-adversarial cut capacity $C$ to 1 unit (after normalization) and append $m$ synthetic Sybil vertices behind a single gateway. Each Sybil vertex has an edge with capacity of 1 unit into the gateway. \Cref{fig:eval-sybil} shows that under both settings, both global policies (i.e., global-count and global-volume) authorize $mC$ for $m$ identities, while \sysname's exposure remains cut-bounded at $C$ regardless of the Sybil count.

\begin{figure}[t]
    \centering
    \includegraphics[width=\columnwidth]{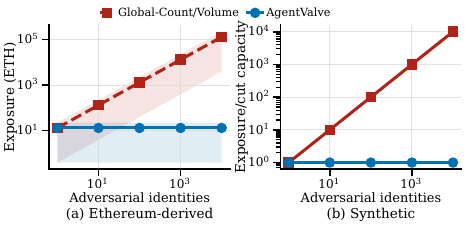}
    \caption{Sybil-region payment exposure under Ethereum-derived and synthetic graphs.
    }
    \label{fig:eval-sybil}
\end{figure}

\begin{figure*}[t]
    \centering
    \includegraphics[width=\textwidth]{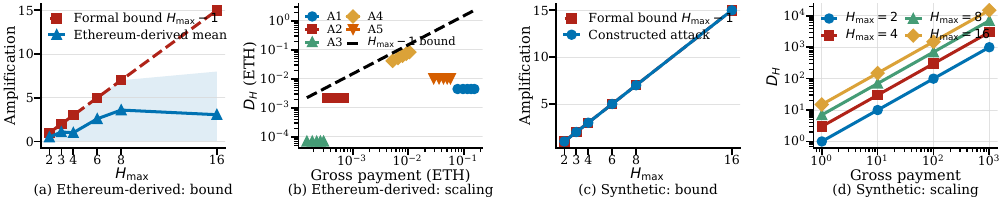}
    \caption{Bounded depletion amplification in Ethereum-derived and synthetic topologies. (a) Selected adversarial senders in the Ethereum-derived graph remain below the exact $H_{\max}-1$ bound across reachable route lengths. (b) Ethereum-derived depleted capacity $D_{\mathbb H}$ versus cumulative adversarial-to-honest payment value $P_{\mathbb A\rightarrow\mathbb H}$. (c) The synthetic worst-case route attains the exact $H_{\max}-1$ bound. (d) Synthetic honest-capacity depletion scales linearly with $P_{\mathbb A\rightarrow\mathbb H}$.}
    \label{fig:eval-depletion}
\end{figure*}

\para{Bounded depletion amplification} (\Cref{fig:eval-depletion}).
We then evaluate bounded depletion amplification, where the metric is calculated as $D_{\mathbb H}/P_{\mathbb A\rightarrow\mathbb H}$, with $D_{\mathbb H}$ as cumulative honest-to-honest route capacity consumed by adversarial senders' negative feedback and $P_{\mathbb A\rightarrow\mathbb H}$ as cumulative adversarial-to-honest payment value.
In the Ethereum-derived setting, we select five vertices $A1$--$A5$ with high-capacity outgoing edges as adversarial senders. 
\Cref{fig:eval-depletion}(a) shows the depletion amplification under different $H_{\max}$ (the shadow reflects the ranges while the solid represents the average), where the adversary picks a simple path with as many edges as possible subject to $H_{\max}$. Mean amplification increases overall from 0.5 at $H_{\max}=2$ to 3.6 at $H_{\max}=8$, then decreases to 3.06 at $H_{\max}=16$. The decrease occurs because the longer path selected at $H_{\max}=16$ can force a much narrower bottleneck (i.e., the smallest capacity across all edges in the path). In other words, the amplification is not necessarily scaled with $H_{\max}$ because the downstream honest-edge capacity can be the bottleneck and limit the amplification.
\Cref{fig:eval-depletion}(b) shows how depleted honest capacity changes as the adversary increases its payment $P_{\mathbb A\rightarrow\mathbb H}$ with $H_{\max}=16$. For A4, the adversarial ingress remains the path bottleneck, so the depleted capacity increases proportionally with the payment. For A1, A2, A3, and A5, however, a downstream honest edge is already the bottleneck. Increasing the payment, therefore, cannot increase the routed amount, causing the depleted capacity to plateau. These results show that additional adversarial expenditure causes additional depletion only when the adversarial ingress is the bottleneck, whereas downstream honest-edge capacity imposes a tighter limit.

\Cref{fig:eval-depletion}(c) shows the depletion amplification under different $H_{\max}$ in the synthetic setting. We construct a path containing exactly $H_{\max}$ edges, where the first edge provides the adversarial ingress and the remaining $H_{\max}-1$ edges originate from honest agents. All edges have capacity $x=P_{\mathbb A\rightarrow\mathbb H}$. The attack therefore depletes $x$ units from each of the $H_{\max}-1$ honest-sourced edges, yielding an amplification of exactly $H_{\max}-1$. The measured amplification increases from 1 at $H_{\max}=2$ to 15 at $H_{\max}=16$, matching the formal upper bound throughout. This result demonstrates that the $(H_{\max}-1)$-bounded depletion-amplification guarantee is tight.
\Cref{fig:eval-depletion}(d) shows how depleted honest capacity changes with $P_{\mathbb A\rightarrow\mathbb H}$ in the synthetic setting. Since every edge on the constructed path has capacity $P_{\mathbb A\rightarrow\mathbb H}$, no downstream edge imposes a narrower bottleneck. Consequently, each additional unit of adversarial payment enables one additional unit of depletion on each of the $H_{\max}-1$ honest-sourced edges, giving $D_{\mathbb H}=(H_{\max}-1)P_{\mathbb A\rightarrow\mathbb H}$. Therefore, the depleted capacity grows linearly with payment. This result confirms that $H_{\max}$ directly bounds the rate at which adversarial expenditure can be amplified into honest-capacity depletion.

\subsection{On-Chain Execution Cost} \label{sec-eval-on-chain-cost}

We next evaluate \sysname's on-chain efficiency from two perspectives: the cost of processing routed payments under different route shapes and the cost of other protocol operations. We report these costs in EVM gas, the unit used to meter contract computation, memory, storage, and calldata, which isolates execution overhead from fluctuations in the gas price and ETH exchange rate. For intuition, at an Ethereum mainnet gas price of 0.1--0.5 gwei, one million gas costs $10^{-4}$--$5\times10^{-4}$ ETH, or approximately \$0.25--\$1.25 at an exchange rate of \$2,498 per ETH.

\begin{figure*}[t]
    \centering
    \includegraphics[width=\textwidth]{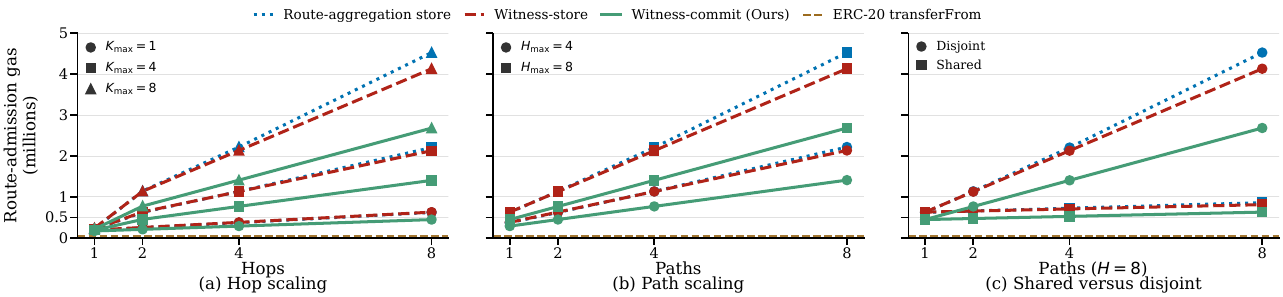}
    \caption{Gas consumed by routed-payment admission under \sysname and the comparison mechanisms. (a) Hop-count scaling with one, four, and eight paths. (b) Path-count scaling for edge-disjoint routes of four and eight hops per path. (c) Shared versus edge-disjoint routes with eight hops per path. Horizontal lines show the route-free ERC-20 \texttt{transferFrom} (53,682 gas).}
    \label{fig:eval-route-admission-gas}
\end{figure*}

\para{Hop-count scaling} (\Cref{fig:eval-route-admission-gas}(a)). For a one-path route, witness-commit admission consumes 168,741 gas at one hop, 287,923 gas at four hops, and 446,844 gas at eight hops. These costs are 12.2\%, 24.1\%, and 28.9\% lower than witness-store, respectively, and 9.7\%, 22.4\%, and 28.3\% lower than route-aggregation store. For an eight-path route, witness-commit saves 10.0\%, 33.9\%, and 35.1\% over witness-store and 4.8\%, 36.4\%, and 40.8\% over route-aggregation store at one, four, and eight hops. The experimental results show that witness-commit consistently reduces gas consumption relative to the two baselines, and the savings grow with route complexity. In addition, the route-aggregation store mechanism is less costly than witness-store when the path count is small (e.g., one path), because its repeated array searches remain inexpensive, while witness-store pays for the larger witness and its verification. However, as the path count grows, the repeated searches become more expensive, and witness-store becomes cheaper than route-aggregation store.

\para{Path-count and overlap scaling} (\Cref{fig:eval-route-admission-gas}(b) and (c)). We compare edge-disjoint routes, whose paths share no edges and therefore have $U=L$, with fully shared routes, whose paths reuse the same eight edges and therefore keep $U=8$ while $L$ grows. For edge-disjoint eight-hop paths (\Cref{fig:eval-route-admission-gas}(b)), increasing the path count from one to two, four, and eight increases $(L,U)$ from $(8,8)$ to $(16,16)$, $(32,32)$, and $(64,64)$. Witness-commit admission consumes 447,168, 766,853, 1,406,238, and 2,685,032 gas across these settings. At eight paths, it saves 35.1\% over witness-store's 4,134,189 gas and 40.7\% over route-aggregation store's 4,528,929 gas. For fully shared routes \Cref{fig:eval-route-admission-gas}(c), increasing the path count from one to eight raises witness-commit admission from 447,168 to 629,727 gas. At eight paths, witness-store consumes 811,669 gas and route-aggregation store consumes 861,649 gas, so witness-commit saves 22.4\% and 26.9\%, respectively. The larger savings for disjoint routes arise because $U$ grows with the path count and witness-commit eliminates the $O(U)$ persistent receipt storage incurred by witness-store.

We also evaluate the additional gas imposed by trust-network routing over the route-free ERC-20 \texttt{transferFrom}~\cite{erc20}, the standard token-transfer operation used on Ethereum. These figures cover payment admission, including the token transfer and capacity reservation; feedback and timeout-resolution costs are reported separately in \Cref{tab:eval-lifecycle-gas}.
As shown in \Cref{fig:eval-route-admission-gas}, witness-commit admission consumes 168,741 gas for a one-hop, one-path route, which is 3.14$\times$ the ERC-20 baseline and approximately \$0.04--\$0.21 under the prices above. The ratio rises to 8.32$\times$ for an eight-hop, one-path route. Repeated counterparties are common in our Ethereum-derived trace: 10,899 of 12,527 retained transfers (87.0\%) occur after an earlier transfer by the same ordered owner pair (\Cref{sec-dataset-detail}). Once a positively resolved interaction establishes sufficient direct-edge capacity, subsequent payments to that provider can use a one-hop, one-path route. The one-hop result thus represents an important common case.

\para{Operation costs}.
\Cref{tab:eval-lifecycle-gas} summarizes the gas consumed by \scname's operations apart from route payments. Agent registration consumes 100,173 gas, which is a one-time cost when an agent is first registered. For trust bootstrap, a provider consumes 121,676 gas to initialize an incoming edge.
TrustRoot-supported seed-payment admission consumes 161,257 gas, while positive sender feedback that releases the TrustRoot reservation and creates the sender's direct edge consumes 96,531 gas. Routed sender and receiver feedback consume 48,813--59,644 and 39,628--48,094 gas, respectively. The negative receiver feedback consumes more gas than the positive and neutral receiver feedback, because it requires the capacity update.

For the timeout recovery cost, we execute several withholding scenarios. In particular, for routed receipts, recovering from missing sender or receiver feedback consumes 45,439 or 34,714 gas. When both roles withhold feedback, two transactions consume 78,588 gas in total. Every routed receipt finalizes, and every reservation reaches zero. For a failed seed payment, sender timeout consumes 43,629 gas, restores the 100 reserved units to the TrustRoot edge, and creates no direct sender-to-provider capacity. After positive sender feedback has activated a successful seed payment, receiver timeout consumes 34,714 gas and preserves all 100 units of available direct capacity.
Moreover, we also explore the efficiency of batch processing. One bounded batch resolves eight withheld roles across four receipts for 123,123 gas, or approximately 15,390 gas per role.
The timeout cases therefore exercise the receipt-specific neutral transitions and confirm that an unrelated caller can restore the intended state. 

\begin{table}[t]
    \setlength\abovecaptionskip{-0.1\baselineskip}
    \centering
    \caption{Gas for \scname operations and timeout recovery.}
    \label{tab:eval-lifecycle-gas}
    \footnotesize
    \setlength{\tabcolsep}{3pt}
    \renewcommand{\arraystretch}{1.04}
    \begin{tabular}{@{}p{0.58\columnwidth}>{\raggedleft\arraybackslash}p{0.36\columnwidth}@{}}
        \toprule
        \textbf{Operation or scenario} & \textbf{Gas used} \\
        \midrule
        \multicolumn{2}{@{}l}{\emph{Normal execution}} \\
        Agent registration & 100,173 \\
        Provider bootstrap & 121,676 \\
        Seed admission / activation & 161,257 / 96,531 \\
        Routed sender feedback (+ / 0 / $-$) & 59,644 / 50,362 / 48,813 \\
        Routed receiver feedback (+ / 0 / $-$) & 39,628 / 39,640 / 48,094 \\
        \addlinespace[2pt]
        \multicolumn{2}{@{}l}{\emph{Permissionless timeout recovery}} \\
        Routed sender / receiver withheld & 45,439 / 34,714 \\
        Routed both withheld (2 calls) & 78,588 \\
        Seed: sender / post-activation receiver & 43,629 / 34,714 \\
        Four routed receipts, eight roles (batch) & 123,123 (15,390/role) \\
        \bottomrule
    \end{tabular}
\end{table}

\begin{figure*}[t]
    \centering
    \includegraphics[width=\textwidth]{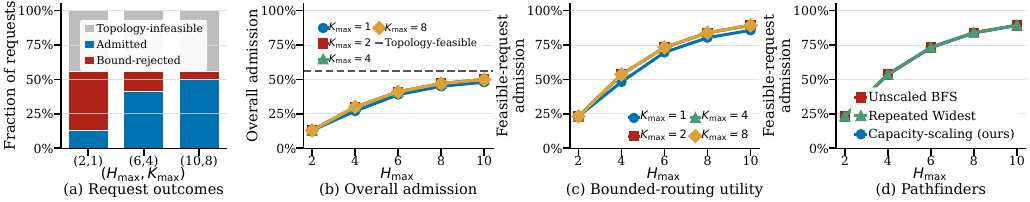}
    \caption{Routing utility on the Ethereum-derived graph. (a) Overall request outcomes at three representative $(H_{\max},K_{\max})$ settings. (b) Overall admission and the topology-feasible ceiling. (c) Admission among topology-feasible requests. (d) Topology-feasible admission for the three pathfinder algorithms at $K_{\max}=4$.
    }
    \label{fig:eval-routing-observed}
\end{figure*}

\begin{figure*}[t]
    \centering
    \includegraphics[width=\textwidth]{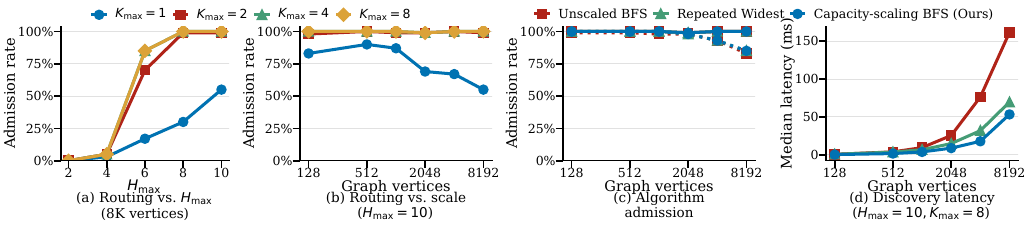}
    \caption{Routing utility and pathfinder scaling on synthetic graphs. (a) Admission versus $H_{\max}$ on the 8,192-vertex graph. (b) Admission versus graph size at $H_{\max}=10$. (c) Algorithm admission at $(H_{\max},K_{\max})=(10,8)$ (solid) and $(6,4)$ (dotted). (d) Median path-discovery latency at $(H_{\max},K_{\max})=(10,8)$.
    }
    \label{fig:eval-routing-pathfinder}
\end{figure*}

\subsection{Routing Utility and Pathfinder Scalability} \label{sec-eval-route}

We finally evaluate how the protocol bounds affect routing utility and how the pathfinder scales with the trust-network size. For each source $s$, the client sets the routing value to $W(s)=\min\{W_0,\max_{(s,v)\in\mathbb{E}_t}a_t(s,v)\}$ using only the available capacities of its outgoing edges, where $W_0$ is the median positive edge capacity (0.00494~ETH) on the Ethereum-derived graph and 50 units on each synthetic graph. We report two complementary utility metrics. \emph{Overall admission} is the fraction of sampled requests for which a pathfinder finds a route carrying the full payment value within $H_{\max}$ hops per path and $K_{\max}$ paths. After fixing each request, we compute its unrestricted maximum flow solely as an offline feasibility label. \emph{Topology-feasible admission} is the admitted fraction among requests whose unrestricted maximum flow can carry $W(s)$ and therefore isolates losses caused by the protocol bounds from requests that no routing algorithm could admit. Each evaluation uses 100 source--destination requests, and we also report path-discovery latency. We present results from the Ethereum-derived graph in \Cref{fig:eval-routing-observed}, followed by synthetic results in \Cref{fig:eval-routing-pathfinder}. The Ethereum evaluation reflects the current state of the network, while the synthetic evaluation simulates a more mature market with a larger and denser network.

\para{Evaluating the Ethereum-derived graph}.
\Cref{fig:eval-routing-observed}(a) separates each request into admitted, bound-rejected (i.e., topology-feasible but rejected by the pathfinder algorithm under the route bounds), and topology-infeasible (whose unrestricted maximum flow is below $W(s)$) outcomes. Although every sampled source--destination pair is connected, 44\% of requests have an unrestricted maximum flow below $W(s)$ and therefore are not feasible regardless of the route bounds. Moreover, tight bounds $(H_{\max},K_{\max})=(2,1)$ admit 13\% of all requests and reject another 43\% that are topology-feasible. At $(6,4)$ and $(10,8)$, admission rises to 41\% and 50\%, while bound-induced rejection falls to 15\% and 6\%, respectively.
This decomposition allows us to isolate the rejected requests that the topology cannot carry and to focus on the effect of the route bounds on the routing utility.

\Cref{fig:eval-routing-observed}(b) and (c) show the corresponding overall and topology-feasible admission rates. At $H_{\max}=10$, increasing $K_{\max}$ from one to two raises overall admission from 48\% to 50\% and topology-feasible admission from 85.7\% to 89.3\%. Larger path budgets provide no additional admission. At $H_{\max}=6$, the same increase raises topology-feasible admission from 69.6\% to 73.2\%. Thus, most requests that the current topology can carry are admitted with $K_{\max}=2$, while increasing the hop budget accounts for the larger utility gains.

\Cref{fig:eval-routing-observed}(d) compares the pathfinders at $K_{\max}=4$. All three achieve topology-feasible admission rates of 23.2\%, 53.6\%, 73.2\%, 83.9\%, and 89.3\% as $H_{\max}$ increases from 2 to 10. The identical admission rates show that the route bounds, rather than the pathfinding heuristic, dominate the routing utility on the sparse Ethereum-derived graph.

Moreover, we can see that since this market is in its infancy, agents will need larger bounds to get their transactions admitted under the sparse network. Fortunately, such high costs will be mitigated as the market grows, which will be shown in the next synthetic evaluation.

\para{Simulating a market with synthetic graphs}.
\Cref{fig:eval-routing-pathfinder}(a) shows the routing utility on the 8,192-vertex synthetic graph. With $K_{\max}=1$, admission rises gradually from 0\% at $H_{\max}=2$ and 3\% at $H_{\max}=4$ to 17\%, 30\%, and 55\% at $H_{\max}=6$, 8, and 10. When $K_{\max}$ increases to 2, admission rises to 70\% at $H_{\max}=6$ and 99\% at $H_{\max}=8$ and 10, showing that allowing one additional path recovers most of the routing utility. With $K_{\max}=8$, admission instead reaches 5\% at $H_{\max}=4$, 85\% at $H_{\max}=6$, and 100\% at $H_{\max}=8$ and 10. In particular, increasing the path budget raises admission at $H_{\max}=6$ from 17\% to 85\%. A single path often cannot carry the full payment even when the destination is reachable within the hop bound, whereas multipath routing can aggregate capacity across several routes. These results show that both bounds affect routing utility.

\Cref{fig:eval-routing-pathfinder}(b) shows the routing admission rates across different graph sizes and path counts with $H_{\max}=10$. At $H_{\max}=10$, a single-path limit admits 83\% of requests on the 128-vertex graph and 55\% on the 8,192-vertex graph. Allowing two paths maintains 98--100\% admission across all graph sizes, while $K_{\max}\geq4$ maintains 99--100\% admission. Thus, although a single path becomes increasingly restrictive as the trust network grows, a four-path budget recovers nearly full admission across the evaluated synthetic graphs.

\Cref{fig:eval-routing-pathfinder}(c) evaluates the admission rates of three pathfinding algorithms under varying trust-network sizes.
At $(H_{\max},K_{\max})=(10,8)$, all three algorithms admit 100\% of requests on the largest synthetic graph. At $(H_{\max},K_{\max})=(6,4)$, the capacity-scaling BFS multipath and repeated widest-path algorithms admit 85\%, while the unscaled BFS multipath algorithm admits 83\%. These results indicate that the route bounds dominate admission at the evaluated settings, while the choice of pathfinding algorithm has only a small effect on routing utility.

\Cref{fig:eval-routing-pathfinder}(d) shows the path-discovery latency of different pathfinding algorithms under distinct trust-network sizes at $(H_{\max},K_{\max})=(10,8)$.
On the 2,048-vertex graph, our capacity-scaling BFS multipath algorithm has a median path-discovery latency of 9.1~ms, which is 65.7\% lower than the 26.5~ms of the unscaled BFS multipath algorithm and 38.5\% lower than the 14.8~ms of the repeated widest-path algorithm. On the 8,192-vertex graph, our capacity-scaling BFS multipath algorithm has a median path-discovery latency of 62.5~ms, which is 61.9\% lower than the 164.0~ms of the unscaled BFS multipath algorithm and 13.0\% lower than the 71.8~ms of the repeated widest-path algorithm. The results show that our capacity-scaling BFS multipath algorithm consistently outperforms the other two algorithms in path-discovery latency, and the performance gap grows with the trust-network size.

\section{Conclusion and Discussion} \label{sec-conclusion}

We presented \sysname, a protocol for Sybil-resilient payment authorization in permissionless AI markets. Its exposure bound applies collectively to all adversarial identities and cannot be enlarged through reputation manipulation confined to the adversarial region. By combining atomic route reservation, interaction-backed capacity creation, and bounded route length, \sysname provides cut-bounded exposure and limits capacity depletion by malicious agents. Its compute--verify decomposition and witness-commit mechanism further enable linear route verification and constant persistent state per admitted payment. Our prototype evaluation shows that these guarantees remain effective under large-scale Sybil attacks while incurring practical on-chain costs.

\para{Depletion accounting.}
The factor $(H_{\max}-1)$ remains a valid, conservative upper bound. Under cumulative payment accounting, negatively resolved ordinary routed payments give factors of $H_{\max}-2$ for adversarial receivers and $(H_{\max}-1)/2$ for honest receivers. The latter case counts both the earlier payments that created ingress capacity and the subsequent negatively resolved payments. Negatively resolved seed payments give a factor of $1$. Since $H_{\max}\geq2$, the largest of these factors simplifies to $\max\{1,H_{\max}-2\}$, yielding a tighter general bound. For simplicity, we retain the factor $(H_{\max}-1)$ as a conservative upper bound throughout the paper.

\para{Exposure over time.}
Cut-bounded exposure bounds a snapshot of risk, not cumulative fraud loss. Provider grants and positive feedback can enlarge the cut. Moreover, it is possible that buyers cannot recognize and report service failures before the deadline; neutral timeout resolution can then restore capacity despite fraud, allowing repeated losses. Increasing the feedback timeout can reduce this risk but also risky in withholding capacity for longer periods. A sound solution may require an adaptive feedback timeout mechanism that allows using different timeouts for different services.

\para{Economic scope.}
The depletion denominator counts gross payments recorded by the protocol. A refund can return funds without removing the capacity they created. Thus, the bound does not establish a lower bound on net attack cost, whose estimation would require additional assumptions about service value, refunds, transaction fees, and opportunity costs, which are beyond the scope of this work and could be an interesting direction for future research.

\para{Reservation griefing.}
Atomic reservation prevents concurrent payments from reusing the same trust capacity, but it also creates a denial-of-service vector: a malicious buyer can keep an admitted route reserved until the sender-feedback timeout and then reserve it again after neutral resolution. Repeatedly doing so can deny routing capacity to honest buyers, constituting a form of reservation griefing analogous to congestion attacks in payment-channel networks~\cite{lu2020general}.
Prior work suggests several possible mitigations, including shorter routes and limits on concurrent requests from less-trusted peers~\cite{mizrahi2021congestion}, as well as capacity-aware routing that avoids congested edges when alternative paths are available~\cite{wang2019flash,sivaraman2020high}. In \sysname, bounding both the path length and the number of paths limits the amount of capacity that a single payment can withhold.
A complementary direction is to use geometrically decreasing reservations toward the receiver. For example, a payment of value 4 over a three-edge path could reserve 16, 8, and 4 units on successive edges. This design could limit downstream lockup when capacity near the buyer is scarce, but at the cost of requiring substantially more upstream capacity. Adapting these mechanisms to further mitigate reservation griefing remains future work.

\ifpublish
\section*{Acknowledgments}
We thank Alberto Sonnino for his helpful feedback during the early stages of this work. Jianting Zhang was supported in part by the Bilsland Fellowship.
\fi

\bibliographystyle{plain}
\bibliography{cite}

\appendix

\balance

\section{Security and Complexity Analysis} \label[appendix]{sec-analysis}

This section proves that \sysname provides cut-bounded exposure, bounded depletion amplification, route soundness, and eventual resolution, and then shows that it meets the efficiency goal in \Cref{sec-property-goal}.

\subsection{Commitment and Capacity Accounting}

We first establish the commitment and capacity-accounting properties used by the subsequent proofs.


\begin{lemma}[Witness-commit binding and recovery] \label{lem:witness-commit}
Let $\rho$ be a routed receipt admitted with route witness $\Lambda(\mathcal{P})$. A witness \emph{opens} $\rho$ if the hash of its encoding matches the commitment stored in $\rho$. 
In \sysname, any witness accepted to open $\rho$ equals $\Lambda(\mathcal{P})$, and any party can reconstruct that witness independently of the pathfinder.
\end{lemma}

\begin{proof}
At admission, \scname verifies the submitted witness against $\mathcal{P}$ and stores a commitment to its encoding in receipt $\rho$. Because the encoding is canonical, accepting a different witness for the same commitment would yield a collision in $\mathsf{Hash}$. The admitted witness appears in the finalized admission calldata and load events, so any chain-data reader can reconstruct and submit it during sender feedback or timeout resolution.
\end{proof}

\begin{lemma}[Capacity-accounting invariant] \label{lem:capacity-accounting}
For every edge $e$ and finalized time $t$, $a_t(e)\geq 0$, $r_t(e)\geq 0$, and $c_t(e)=a_t(e)+r_t(e)$. Reserved capacity $r_t(e)$ equals the sum of admitted ordinary route loads and direct TrustRoot-edge loads on $e$ whose sender roles remain unresolved. Admission of either payment type atomically moves its supporting loads from available to reserved capacity without changing total capacity. Every sender-role resolution removes its reservation exactly once, either returning it to available capacity or consuming it from total capacity.
\end{lemma}

\begin{proof}
The claim follows by induction over the totally ordered contract transitions. Setup initializes non-negative available capacities and zero reservations. A successful provider join adds $W$ to the available and total capacity of $e_{r,d}$, with $r$ a TrustRoot, while leaving its reserved capacity unchanged (\Cref{sec-protocol-bootstrap}). 
Before admitting a routed payment, \scname checks $\ell_e(\mathcal{P})\leq a_t(e)$ for every route edge and then atomically applies $a_{t+1}(e)=a_t(e)-\ell_e(\mathcal{P})$ and $r_{t+1}(e)=r_t(e)+\ell_e(\mathcal{P})$, preserving nonnegativity and total capacity. Bootstrap admission applies the same transition with load $W$ on $e_{r,d}$ after checking $a_t(r,d)\geq W$. A later routed sender-role resolution opens the receipt to exactly its admitted loads by \Cref{lem:witness-commit}; a bootstrap resolution obtains its single edge and load from immutable receipt metadata. Positive, neutral, and timeout resolutions release each load by decreasing $r_t(e)$ and increasing $a_t(e)$ by the same amount, while a negative resolution decreases both $r_t(e)$ and $c_t(e)$ and leaves $a_t(e)$ unchanged. For either payment type, positive sender feedback additionally increases $a_t(s,d)$ and $c_t(s,d)$ by $W$ on the sender-to-receiver edge. Receipt-type and role-finalization checks ensure that exactly one applicable transition is executed. Receiver feedback either preserves capacity or decreases available and total capacity by the same nonnegative amount. All other transitions leave the invariant unchanged.
\end{proof}

\subsection{Exposure and Depletion Bounds}

The security bounds use the following relationship between a route and the directed cuts it crosses. Recall that $\mathbb{H}$ is the honest set, $\mathbb{A}$ is the adversarial set, and their direct cut is denoted by $\delta_t(\mathbb{H},\mathbb{A})$.

\begin{lemma}[Cut-crossing loads] \label{lem:cut-crossing}
Fix a time $t$, a partition $(\mathbb{H},\mathbb{A})$, and a route $\mathcal{P}=\{(P_i,w_i)\}_{i=1}^k$ in $\mathbb{G}_t$ with value $W$. If its sender is in $\mathbb{H}$ and its receiver is in $\mathbb{A}$, then
\[
    \sum_{e\in\delta_t(\mathbb{H},\mathbb{A})}\ell_e(\mathcal{P})\geq W.
\]
If its sender is in $\mathbb{A}$, then the load on honest-to-honest edges is bounded by
\[
    \sum_{e=(u,v):u,v\in\mathbb{H}}\ell_e(\mathcal{P})
    \leq
    (H_{\max}-1)
    \sum_{e\in\delta_t(\mathbb{A},\mathbb{H})}\ell_e(\mathcal{P}).
\]
\end{lemma}

\begin{proof}
For the first inequality, every path from $\mathbb{H}$ to $\mathbb{A}$ crosses the cut $\delta_t(\mathbb{H},\mathbb{A})$ at least once, so path $P_i$ contributes at least $w_i$ to the cut load. Summing gives at least $\sum_iw_i=W$. For the second inequality, let $h_i$ be the number of edges in $P_i$ whose two endpoints are honest, and let $q_i$ be its number of $\mathbb{A}$-to-$\mathbb{H}$ crossings. These are disjoint sets of edges. The bounded-stretch routing rule gives $h_i+q_i\leq |P_i|\leq H_{\max}$. If $h_i>0$, then $q_i\geq1$ because $P_i$ begins in $\mathbb{A}$, so $h_i\leq H_{\max}-q_i\leq(H_{\max}-1)q_i$. The inequality is immediate when $h_i=0$. 
Multiplying by $w_i$, summing over paths, and exchanging the path and edge sums yields the result.
\end{proof}

The following theorem adapts Bazaar's cut argument~\cite{post2011bazaar} to \sysname's directed trust graph and concurrent contract execution. Atomic reservation ensures that the same unit of cut capacity cannot support multiple payments.
Note that the exposure $F_{\mathbb H}(t)$ is state-local: it covers unresolved payments and additional payments jointly admissible from state $t$, but excludes losses from receipts already resolved before $t$.

\begin{theorem}[Cut-bounded exposure] \label{thm:cut-bounded-exposure}
\sysname satisfies cut-bounded exposure as defined in \Cref{def:cut-bounded-exposure}.
\end{theorem}

\begin{proof}
Fix a time $t$. By \Cref{lem:cut-crossing}, every ordinary routed payment from an honest sender to an adversarial receiver places at least its value in aggregate load on $\delta_t(\mathbb{H},\mathbb{A})$. A bootstrap payment to an adversarial receiver $d$ reserves its value $W$ on $e_{r,d}$, which crosses the same cut because the TrustRoot $r$ is honest. Provider joins to adversarial recipients can add TrustRoot-sourced cut capacity; any grant still available or reserved at $t$ is counted by $a_t(e)+r_t(e)$. The cut loads supporting admitted payments whose sender roles remain unresolved are included in $r_t(e)$. Any additional payment exposure jointly admissible from state $t$ must fit within $a_t(e)$ on every cut edge. Atomic admission prevents concurrent payments of either type from reusing available capacity, while \Cref{lem:capacity-accounting} prevents reserved capacity from supporting another payment before resolution. The exposure from payments whose sender roles remain unresolved at $t$ and further payments jointly admissible from state $t$ is therefore at most the sum of available and reserved capacity across the cut; losses from receipts resolved before $t$ are excluded. Hence $F_{\mathbb H}(t)\leq\sum_{e\in\delta_t(\mathbb H,\mathbb A)}(a_t(e)+r_t(e))=C_t(\mathbb H,\mathbb A)$.
\end{proof}

\begin{corollary}[Sybil invariance] \label{cor:sybil-invariance}
Holding the honest-to-adversarial cut capacity fixed, adding adversarial identities or changing edges and capacities entirely within $\mathbb{A}$ leaves the aggregate exposure bound in \Cref{thm:cut-bounded-exposure} unchanged.
\end{corollary}

\begin{proof}
These changes do not alter the capacity of $\delta_t(\mathbb{H},\mathbb{A})$. The bound $C_t(\mathbb{H},\mathbb{A})$ therefore remains unchanged and applies jointly to all adversarial receivers, regardless of their number or internal trust structure.
\end{proof}

\Cref{cor:sybil-invariance} implies \sysname's Sybil-resilience. Specifically, creating identities or collusive interactions within an adversarial group cannot increase the bound on honest agents' aggregate payment exposure, which is determined by the trust capacity they extend to that group.

An adversarial provider can increase the exposure bound through an explicit TrustRoot grant (i.e., via the provider joining process, described in \Cref{sec-protocol-bootstrap}), but this is not free and is under control. A grant of $W$ charges an irrecoverable fee of $W$ from the adversarial provider, adds $W$ to the honest-to-adversarial cut capacity, and deducts $W$ from the TrustRoot's preauthorized grant budget. The system can make the fees irrecoverable by burning them, and the TrustRoot only replenishes its grant budget periodically, which controls the rate at which new capacity can be added.
As a result, splitting a total joining fee across Sybil identities cannot obtain more bootstrap capacity than that total fee permits. Together with \Cref{cor:sybil-invariance}, creating Sybil identities along cannot bring the adversary more benefits in defrauding honest agents without paying the same amount of fees.



\begin{lemma}[Bootstrap depletion] \label{lem:bootstrap-depletion}
Let $D^{\mathrm{boot}}_{\mathbb H}(t)$ denote the cumulative capacity consumed by time $t$ on honest-to-honest edges by negative sender feedback on seed payments initiated in $\mathbb A$, and let $P^{\mathrm{boot}}_{\mathbb A\rightarrow\mathbb H}(t)$ denote the cumulative value of seed payments from adversarial senders to honest receivers by $t$. Seed payments are admitted through $\network.\text{seedPay}$ (\Cref{sec-protocol-bootstrap}). Then
\[
    D^{\mathrm{boot}}_{\mathbb H}(t)
    \leq P^{\mathrm{boot}}_{\mathbb A\rightarrow\mathbb H}(t).
\]
\end{lemma}

\begin{proof}
A negatively resolved seed payment of value $W$ initiated in $\mathbb A$ consumes only $W$ on $e_{r,d}$. Since $r$ is honest, this is honest-to-honest depletion exactly when $d$ is honest; in that case, its adversarial sender has transferred $W$ to an honest receiver at admission. If $d$ is adversarial, the consumed edge is outside the measured honest-to-honest subgraph. Let $Q(t)$ denote the cumulative value of seed payments from $\mathbb A$ to $\mathbb H$ whose sender roles have resolved negatively by $t$. Each sender role resolves at most once. Only negative sender feedback consumes the reserved capacity; positive feedback, neutral feedback, and timeout resolution release it back to available capacity. Therefore, $D^{\mathrm{boot}}_{\mathbb H}(t)=Q(t)$. Moreover, every payment counted by $Q(t)$ was admitted by $t$ and is therefore also counted in $P^{\mathrm{boot}}_{\mathbb A\rightarrow\mathbb H}(t)$. Hence $D^{\mathrm{boot}}_{\mathbb H}(t)=Q(t)\leq P^{\mathrm{boot}}_{\mathbb A\rightarrow\mathbb H}(t)$, proving the bound.
\end{proof}

\begin{theorem}[Bounded depletion amplification] \label{thm:bounded-depletion}
\sysname provides $(H_{\max}-1)$-bounded depletion amplification as defined in \Cref{def:bounded-depletion}.
\end{theorem}

\begin{proof}
Fix a time $t$ and the honest and adversarial sets $(\mathbb H,\mathbb A)$ of an execution under \Cref{sec-adversary-model}. Let $D^{\mathrm{route}}_{\mathbb H}(t)$ denote the cumulative capacity consumed by time $t$ on honest-to-honest edges by negative sender feedback on ordinary routed payments initiated in $\mathbb A$ and admitted through $\network.\text{transact}$. Together with the seed-payment contribution $D^{\mathrm{boot}}_{\mathbb H}(t)$ defined in \Cref{lem:bootstrap-depletion}, this gives $D_{\mathbb H}(t)=D^{\mathrm{route}}_{\mathbb H}(t)+D^{\mathrm{boot}}_{\mathbb H}(t)$.

For an ordinary routed payment initiated in $\mathbb A$ whose sender role resolves negatively, let $x$ be its aggregate load on edges directed from $\mathbb A$ to $\mathbb H$. By \Cref{lem:cut-crossing}, the route consumes at most $(H_{\max}-1)x$ capacity on honest-to-honest edges. Summed across all such resolutions through $t$, each unit of adversarial-to-honest capacity can be permanently consumed at most once by \Cref{lem:capacity-accounting}.

Let $P^{\mathrm{pos}}_{\mathbb A\rightarrow\mathbb H}(t)$ denote the cumulative value of ordinary routed payments and seed payments from adversarial senders to honest receivers whose sender roles have resolved positively by $t$. Setup creates no adversarial-to-honest capacity because every initial edge originates at an honest TrustRoot (\Cref{sec-protocol-setup}). Provider joining likewise adds only TrustRoot-sourced capacity (\Cref{sec-protocol-bootstrap}). Thereafter, only positive sender feedback on a payment from $\mathbb A$ to $\mathbb H$ creates adversarial-to-honest capacity, adding exactly the payment value. Cumulative consumption on these edges is therefore at most $P^{\mathrm{pos}}_{\mathbb A\rightarrow\mathbb H}(t)$, giving $D^{\mathrm{route}}_{\mathbb H}(t)\leq(H_{\max}-1)P^{\mathrm{pos}}_{\mathbb A\rightarrow\mathbb H}(t)$.

By \Cref{lem:bootstrap-depletion}, $D^{\mathrm{boot}}_{\mathbb H}(t)=Q(t)$, where $Q(t)$ is the cumulative value of seed payments from $\mathbb A$ to $\mathbb H$ whose sender roles have resolved negatively by $t$. These payments are disjoint from those counted by $P^{\mathrm{pos}}_{\mathbb A\rightarrow\mathbb H}(t)$ because each sender role resolves at most once. Thus, $P^{\mathrm{pos}}_{\mathbb A\rightarrow\mathbb H}(t)+Q(t)\leq P_{\mathbb A\rightarrow\mathbb H}(t)$. Since $H_{\max}-1\geq1$,
\[
    \begin{aligned}
        D_{\mathbb H}(t)
        &\leq (H_{\max}-1)P^{\mathrm{pos}}_{\mathbb A\rightarrow\mathbb H}(t)+Q(t) \\
        &\leq (H_{\max}-1)\bigl(P^{\mathrm{pos}}_{\mathbb A\rightarrow\mathbb H}(t)+Q(t)\bigr) \\
        &\leq (H_{\max}-1)P_{\mathbb A\rightarrow\mathbb H}(t).
    \end{aligned}
\]
\end{proof}

The quantity $P_{\mathbb A\rightarrow\mathbb H}(t)$ records the cumulative value transferred through the protocol from adversarial agents to honest agents up to time $t$, rather than the adversary's net economic loss. Service value, refunds, reciprocal payments, and transfers outside the protocol are not deducted. Quantifying these offsets would require additional assumptions about service valuation and the attribution of transfers to particular interactions. Cumulative accounting ensures that capacity accumulated before a later attack remains charged to the paid interactions that created it.

\subsection{Integrity and Liveness}



\begin{theorem}[Route soundness] \label{thm:route-soundness}
\sysname satisfies route soundness as defined in \Cref{def:route-soundness}.
\end{theorem}

\begin{proof}
For each routed payment, \scname checks the path-count and hop-count bounds, registered endpoints, edge continuity, path simplicity, positive path values, $\sum_iw_i=W$, and the unique-edge bound (\Cref{sec-protocol-transaction}). It independently recomputes the aggregate loads represented by the route witness and checks $\ell_e(\mathcal{P})\leq a_t(e)$ against the contract state immediately before admission for every route edge. A failed check reverts without changing the state. If all checks pass, the atomic transition in \Cref{lem:capacity-accounting} reserves exactly these loads and leaves every $a_{t+1}(e)\geq0$. Therefore, every admitted routed transaction is feasible at its atomic admission point and cannot overbook an edge.
In addition, by \Cref{lem:witness-commit}, subsequent sender feedback or timeout resolution accepts only the admitted route witness and therefore releases or consumes exactly the corresponding reservation. An untrusted or stale pathfinder output can consequently cause admission to fail, but it cannot induce an invalid trust-network transition. These facts establish route soundness as defined in \Cref{def:route-soundness}.
\end{proof}

\begin{theorem}[Eventual resolution] \label{thm:eventual-resolution}
\sysname satisfies eventual resolution as defined in \Cref{def:eventual-resolution}.
\end{theorem}

\begin{proof}
Every admitted routed or seed payment produces a receipt with a feedback deadline set $T_{\mathrm{fb}}$ after admission and separate finalization flags for the sender and receiver roles. Before the deadline, the contract authenticates feedback for the selected role and permits that role to resolve only once. After the deadline, $\network.\text{resolve}(\rho,\eta,\widehat{\Lambda})$ is called by any honest party (\Cref{sec-adversary-model}) for either unresolved role $\eta$ and applies the neutral outcome for that role. If a routed sender role remains unresolved, \Cref{lem:witness-commit} allows the caller to reconstruct and submit its unique route witness, after which the contract returns the corresponding loads to available capacity. If a bootstrap sender role remains unresolved, the contract derives $e_{r,d}$ and $W$ from the receipt and returns that reservation to available capacity without a route witness or capacity creation. Resolving a receiver role requires no route witness and does not affect the sender-governed reservation. Under the maintenance assumption in \Cref{sec-adversary-model}, an honest caller eventually submits every required resolution. Eventual finalization and finalized chain-data availability therefore prevent either a withholding participant or the original pathfinder from preserving a reservation indefinitely.
\end{proof}

\subsection{On-chain Complexity}

Recall that $L(\mathcal{P})=\sum_i|P_i|$ is the total number of edge occurrences in a submitted route, and let $U$ be its number of unique edges.

\begin{proposition}[Verification complexity] \label{prop:verification-complexity}
For every route within the protocol limits, \sysname performs $O(L(\mathcal{P})+U)$ verification work and uses $O(U)$ auxiliary memory. This work is at most $O(K_{\max}H_{\max})$ and is independent of the trust graph's size.
\end{proposition}

\begin{proof}
Let $k$ denote the number of paths in $\mathcal{P}$. Under the witness-commit mechanism (\Cref{sec-route-witness-commit}), the client supplies the sorted edge-load table $\Lambda(\mathcal{P})$ and the index of the table entry corresponding to each path-edge occurrence. \scname scans each path to check endpoints, continuity, hop count, and payment value. It also scans each path's node-order certificate, checking its length, position bounds, and strictly increasing vertex identifiers to establish path simplicity. Since a path with $|P_i|$ edges has $|P_i|+1$ vertices, these checks require $O(k+L(\mathcal{P}))$ work.

During the route scan, \scname verifies that each indexed witness entry names the corresponding edge and adds the path's value directly to that entry's memory counter. This requires constant work per occurrence and $O(U)$ auxiliary memory for the counters. A scan of the $U$ witness records checks their ordering and uniqueness, compares the claimed loads with the recomputed totals, and checks each edge's available capacity. Hashing the witness for the receipt commitment adds $O(U)$ work since each record has a fixed size.

The total work is therefore $O(k+L(\mathcal{P})+U)=O(L(\mathcal{P})+U)$, since every path contains at least one edge and hence $k\leq L(\mathcal{P})$. Moreover, $U\leq L(\mathcal{P})\leq kH_{\max}\leq K_{\max}H_{\max}$. Route discovery takes place off-chain, while on-chain verification processes only the submitted route, witness, and corresponding edge states. Its cost is therefore bounded by protocol parameters and independent of $|\mathbb V_t|$ and $|\mathbb E_t|$.
\end{proof}

\begin{proposition}[Receipt-state complexity] \label{prop:receipt-state-complexity}
For a route containing $U$ unique edges, the witness-commit mechanism uses $O(1)$ persistent receipt storage, while receipt opening and reservation resolution require $O(U)$ witness data and work.
\end{proposition}

\begin{proof}
Under the witness-commit mechanism (\Cref{sec-route-witness-commit}), each routed receipt stores the hash commitment $\mathsf{Hash}(\Lambda(\mathcal{P}))$ and a fixed number of metadata fields. Their sizes are independent of $U$, giving $O(1)$ persistent receipt storage.

To resolve the sender role, the resolver re-supplies $\Lambda(\mathcal{P})$, the edge-load table verified at admission. The table contains $U$ fixed-size records, so the witness data has size $O(U)$. Checking its record count, ordering, encoding, and hash requires $O(U)$ work. By \Cref{lem:witness-commit}, a matching commitment authenticates the admitted loads, so the paths and their verification certificates need not be re-supplied or rechecked. \scname then scans the table to release or consume each recorded load, performing one reservation update per record. These updates modify shared edge state without storing per-receipt edge-load records. Positive sender feedback additionally updates or creates one direct edge, adding only $O(1)$ work and state. As a result, receipt opening and reservation resolution require $O(U)$ work.

Since $U\leq L(\mathcal{P})\leq K_{\max}H_{\max}$, resolution cost is bounded by protocol parameters and independent of the trust graph's size. A seed receipt identifies its single reserved edge and payment value through constant-size metadata. It therefore requires no route witness and uses $O(1)$ persistent receipt storage and $O(1)$ admission and resolution work.
\end{proof}

Together, \Cref{thm:cut-bounded-exposure,cor:sybil-invariance,thm:bounded-depletion,thm:route-soundness,thm:eventual-resolution} establish the security properties in \Cref{sec-property-goal}. \Cref{prop:verification-complexity,prop:receipt-state-complexity} further show that on-chain work is bounded independently of the trust graph's size and that each receipt adds only constant persistent state, meeting the efficiency goal.

\section{Additional Evaluation} \label[appendix]{sec-more-experiments}

This section evaluates topology-sensitive liquidity loss and transaction contention. Recall that the Ethereum-derived graph contains 1,055 vertices and 1,628 directed edges with the median capacity of 0.00494~ETH, while synthetic graphs have six directed edges per vertex and capacities sampled uniformly from $[10,100]$ (\Cref{sec-evaluation}).

\subsection{Reachability under Capacity Depletion}

We first examine how a fixed capacity-depletion budget translates into loss of network reachability. This complements the economic bound by measuring the topological consequences of removing capacity rather than the payment value required to acquire it.

\begin{figure*}[t]
    \centering
    \includegraphics[width=\textwidth]{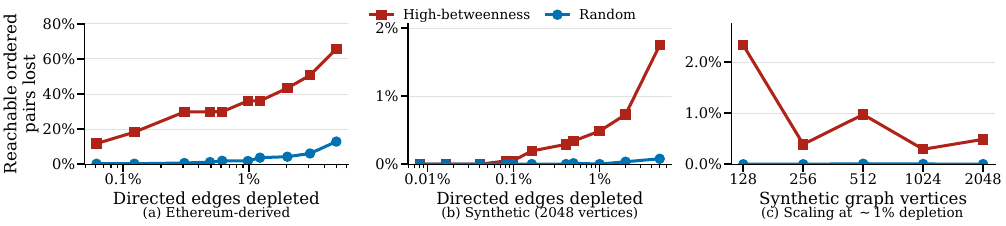}
    \caption{Topology sensitivity under equal unit-capacity depletion budgets. (a) The Ethereum-derived graph. (b) The 2,048-vertex synthetic graph. (c) Synthetic-size scaling after depleting approximately 1\% of directed edges. 
    }
    \label{fig:eval-topology}
\end{figure*}

Our evaluations cover both Ethereum-derived graph and synthetic graph settings. In particular, we assign one unit of capacity to every edge, rank directed edges by unweighted edge-betweenness centrality, and compare removing the highest-ranked edges with uniformly random edge removal under the same number of removed edges. We measure reachability as the number of ordered source--destination pairs connected by a directed path. Each random-removal evaluation involves 20 attempts, while targeted-removal evaluations are deterministic and involve only one attempt.

\Cref{fig:eval-topology} shows the results. The intact Ethereum-derived graph has 20,509 reachable ordered pairs. Removing only the highest-ranked edge eliminates 11.78\% of them, compared with 0.27\% on average under random removal. Removing five targeted edges eliminates 29.81\% versus 0.70\% under random removal. Removing 50 eliminates 50.69\% versus 6.12\%, and removing 81 eliminates 65.78\% versus 12.90\%. After removing approximately 1\% of directed edges, the Ethereum-derived topology loses 36.08\% of reachable pairs under targeting and 1.84\% under random removal, while a 2,048-vertex synthetic graph loses only 0.49\% and 0.0024\%, respectively. On the 2,048-vertex synthetic graph, even a 5\% targeted depletion loses only 1.75\% of reachable pairs, compared with 0.08\% under random removal. Across synthetic sizes, targeted removal at an approximately 1\% budget loses 0.29--2.33\% of reachable pairs, while random removal loses at most 0.01\%. These results show that the Ethereum-derived graph is unusually vulnerable and that targeted removal causes greater reachability loss than random removal across both graph families.

This study is a unit-capacity sensitivity test that removes monetary heterogeneity. Edge removal assumes that the attacker can acquire suitable ingress capacity and submit routes traversing the selected edges, while reachability loss is distinct from lost payment value. The much larger loss on the Ethereum-derived graph indicates that concentration in the current young ecosystem, rather than the mere presence of central edges, drives the most severe damage. Candidate defenses include criticality-aware path scoring, edge-disjoint or low-overlap routes, reserving less of a bridge edge per transaction, and monitoring central edges so honest participants can establish redundant capacity. We leave these as open design problems for future work.

\subsection{Contention under Stale Route Snapshots} \label{sec-eval-contention}

We next examine how delayed pathfinder snapshots interact with serialized on-chain admission. We vary offered load, demand distribution, snapshot delay, transaction value, and client retry behavior to separate stale-route failures from genuine capacity contention. The transaction-value study includes a separate controlled parallel graph with eight independent source--destination groups, each connected by eight parallel two-hop branches.

\begin{figure}[t]
    \centering
    \includegraphics[width=\columnwidth]{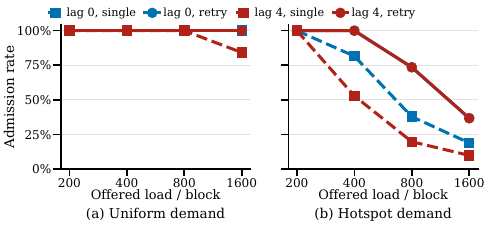}
    \caption{Admission under serialized block contention on the 2,048-vertex synthetic graph. Uniform demand spreads requests across the graph, while hotspot demand concentrates requests on shared branches. 
    }
    \label{fig:eval-contention}
\end{figure}

\begin{figure}[t]
    \centering
    \includegraphics[width=\columnwidth]{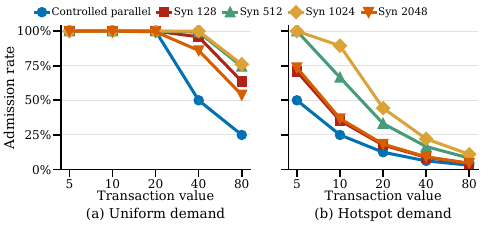}
    \caption{Admission versus transaction value at 1,600 offered requests per block and a four-block snapshot delay under refresh-and-retry. Uniform demand is shown on the left and hotspot demand on the right. 
    }
    \label{fig:eval-transaction-value}
\end{figure}

In the evaluation, requests in the same block use a common pathfinder snapshot, while contract admissions are serialized against current state and reservations remain locked for two blocks. The snapshot delay (or we call it indexer lag), is the gap between the contract state and the state visible to the pathfinder. It is used to simulate the synchronization delay between the on-chain state and the off-chain pathfinder in practice. In \emph{single-shot} mode, a client submits the route once. In \emph{refresh-and-retry} mode, it obtains a fresh route after a rejection and resubmits. Uniform demand spreads requests across source--destination groups, while hotspot demand concentrates them on shared branches. We sweep 200--1,600 offered requests per block, snapshot delays of 0--4 blocks, and transaction values of 5--80 capacity units.

\para{Stale routes and retry} (\Cref{fig:eval-contention}). On the 2,048-vertex synthetic graph, uniform demand at 1,600 offered requests per block and a four-block snapshot delay admits 84.2\% of requests with single-shot submission and 100\% with refresh-and-retry. Hotspot demand is more severe. At the same load and delay, admission rises from 9.9\% to 36.8\% with retry, a 26.9-percentage-point gain, while the remaining failures reflect genuine capacity contention rather than stale-route errors. At 200 offered requests per block, both policies admit all requests, showing that retry helps when concurrent demand approaches shared route capacity. 
These results indicate that high-contention scenarios might lead to admission failures, but retrying with a fresh route can significantly improve the success rate. 

\para{Transaction value} (\Cref{fig:eval-transaction-value}). Increasing the value reserved by each transaction consumes capacity more rapidly. At 1,600 offered requests per block, a four-block snapshot delay, and refresh-and-retry, uniform-demand admission on the 2,048-vertex synthetic graph decreases from 100.0\% at value 5 to 53.9\% at value 80. Under hotspot demand it decreases from 73.6\% to 4.6\% over the same sweep. The controlled parallel graph exhibits the same trend but is more sharply capacity-limited under hotspot demand, falling from 50.0\% to 3.1\%. Thus, retry mitigates stale-route failures but cannot overcome insufficient aggregate capacity, and larger transaction values expose this economic limit directly.

\section{Dataset Collection and Processing} \label[appendix]{sec-dataset-detail}

\para{Frozen real-world dataset.} We collect Ethereum mainnet blocks 24,339,871--25,664,244 inclusive, spanning the ERC-8004 deployment window beginning January~29, 2026 through the frozen August~2, 2026 cutoff. 
The dataset is used to trace interactions of clients, predicting payments performed by AI agents. 
Using Etherscan~\cite{etherscan} log and account APIs, the collector first retrieves 52,322 ERC-721 \texttt{Transfer} events from the ERC-8004 Identity Registry at \nolinkurl{0x8004A169FB4a3325136EB29fA0ceB6D2e539a432}. Ordered mint, transfer, and burn events reconstruct 52,322 ownership intervals. For each of 9,163 identified owners, we retrieve normal and internal transactions and keep an interaction only when the address owned at least one registered agent at the transaction's exact block and transaction position. We deduplicate repeated records by transaction hash, owner, sender, and receiver. 

The raw owner-interaction collection contains 491,079 transactions. We then apply an auditable sequence of filters. Valid-endpoint checking removes 0 records. Excluding registry, zero, and burn addresses removes 33,520. Requiring both endpoints to be agent owners removes 423,052 and leaves 34,507 owner-to-owner candidates. Known failed or reverted transactions remove 217, self-transfers remove 18,598, and zero, negative, or invalid values remove 3,165. No retained transaction has unknown execution status. The resulting directed graph aggregates the count and transferred value of all retained transactions from the same owner to the same owner. Each owner address is one vertex annotated with all ERC-8004 agents it controls, and an edge's initial available capacity is the aggregate transferred value in wei, where $10^{18}$ wei equals one ETH. \Cref{fig:eval-dataset} (right) shows the retained transactions after applying each filter rule. 

This initialization is an optimistic collection that treats every retained historical transfer as though its sender had positively resolved it under \sysname. We refer to the proxy as the \emph{Ethereum-derived owner-transfer graph}, or the \emph{Ethereum-derived graph} for short. We retain edge direction and exclude ERC-8004 feedback because such feedback is not reliably bound to an interaction and is itself susceptible to fictitious activity. 

\begin{figure}[t]
    \centering
    \includegraphics[width=\columnwidth]{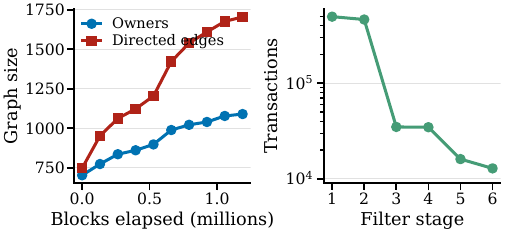}
    \caption{Ethereum-derived dataset. Left: the Ethereum-derived graph continues to grow across the collection window. Right: transactions retained after (1) valid-endpoint checking, (2) exclusion of identity/reputation registries and zero/burn addresses, (3) requiring both endpoints to be agent owners, (4) successful/non-reverted execution, (5) distinct endpoints with no self-transfers, and (6) positive, valid transaction values.}
    \label{fig:eval-dataset}
\end{figure}

\para{Ethereum-derived graph.} The final graph contains 1,055 owner vertices, 1,628 directed edges, and 12,527 retained transactions totaling 169.294~ETH. Grouped by ordered owner pair and chronological transaction order, 10,899 of the 12,527 retained transfers (87.0\%) occur after an earlier transfer by the same pair, while the remaining 1,628 are the first retained transfer for their pair. It has 137 weakly connected components (largest: 213 vertices), 752 strongly connected components (largest: 105), and 20,509 reachable ordered pairs, only 1.844\% of all ordered pairs. The mean in- and out-degree is 1.54, the 95th percentile is 5, and 41.4\% of directed edges have a reverse edge. The median edge aggregates two transactions and 0.00494~ETH, while its 95th and 99th percentile capacities are 0.254 and 1.609~ETH, respectively. Of the identified owners, 611 (57.9\%) control multiple agents. Among retained transactions, 71.6\% touch at least one such owner. \Cref{fig:eval-dataset} (left) shows that the graph grows rather than arising from one short burst, but also makes clear that ERC-8004 is young and the Ethereum-derived topology is sparse.

\end{document}